\documentclass[AMA,Times1COL]{WileyNJDv5} %STIX1COL,STIX2COL,STIXSMALL

\usepackage{acronym}
\usepackage{tikz}
\usetikzlibrary{tikzmark,decorations.pathreplacing,patterns,calc}
\usetikzlibrary{arrows, arrows.meta}
\usepackage{bm}

\usepackage{nth}
\usepackage{subcaption}
\usepackage[font=small]{caption}
\definecolor{amaranth}{rgb}{0.9, 0.17, 0.31}

\usepackage{acronym}

\acrodef{2d}[2D]{bi-dimensional}
\acrodef{5g}[5G]{fifth-generation}
\acrodef{6g}[6G]{sixth-generation}
\acrodef{aoa}[AoA]{angle of arrival}
\acrodef{ao}[AO]{alternating optimization}
\acrodef{aod}[AoD]{angle of departure}
\acrodef{ble}[BLE]{bluetooth low-energy}
\acrodef{bs}[BS]{base station}
\acrodef{cam}[CAM]{cooperative awareness message}
\acrodef{cav}[CAV]{cooperative automated vehicle}
\acrodef{c-its}[C-ITS]{cooperative intelligent transport system}
\acrodef{csi}[CSI]{channel state information}
\acrodef{cp}[CP]{canonical polyadic}
\acrodef{cfar}[CFAR]{constant false alarm rate}
\acrodef{ca-cfar}[CA-CFAR]{cell averaging \ac{cfar}}
\acrodef{ccrb}[CCRB]{constrained Cramér-Rao bound}
\acrodef{cdf}[CDF]{cumulative density function}
\acrodef{crb}[CRB]{Cramér-Rao bound}
\acrodef{dais}[DAIS]{delay-angle information spoofing}
\acrodef{dft}[DFT]{discrete Fourier transform}
\acrodef{ekf}[EKF]{extended Kalman filter}
\acrodef{esprit}[ESPRIT]{estimation of signal parameters via rotational invariance techniques}
\acrodef{ff}[FF]{far-field}
\acrodef{fft}[FFT]{fast Fourier transform}
\acrodef{fim}[FIM]{Fisher information matrix}
\acrodef{fr}[FR]{frequency range}

\acrodef{gnss}[GNSS]{global navigation satellite system}
\acrodef{gospa}[GOSPA]{generalized optimal sub-pattern assignment}
\acrodef{hd}[HD]{high-definition}
\acrodef{hap}[HAP]{high-accuracy positioning}
\acrodef{its}[ITS]{intelligent transportation system}
\acrodef{los}[LoS]{line-of-sight}
\acrodef{ls}[LS]{least squares}

\acrodef{mae}[MAE]{mean absolute error}
\acrodef{mcrb}[MCRB]{mismatched \ac{crb}}
\acrodef{mmWave}[mmWave]{Millimeter-wave}
\acrodef{mimo}[MIMO]{multiple-input multiple-output}
\acrodef{miso}[MISO]{multiple-input single-output}
\acrodef{mpc}[MPC]{multipath component}
\acrodef{ml}[ML]{maximum likelihood}
\acrodef{mf}[MF]{matched-filter}
\acrodef{nlos}[NLoS]{non-line-of-sight}
\acrodef{nf}[NF]{near-field}
\acrodef{nn}[NN]{nearest-neighbour}
\acrodef{ofdm}[OFDM]{orthogonal frequency division multiplexing}
\acrodef{omp}[OMP]{orthogonal matching pursuit}
\acrodef{peb}[PEB]{position error bound}
\acrodef{pla}[PLA]{physical layer authentication}
\acrodef{poam}[POAM]{penalized optimal assignment metric}
\acrodef{prmse}[PRMSE]{penalized root mean squared error}
\acrodef{prs}[PRS]{positioning reference signal}
\acrodef{roi}[ROI]{region of interest}
\acrodef{rf}[RF]{radio frequency}
\acrodef{rfc}[RFC]{radio frequency chain}
\acrodef{rrm}[RRM]{radio-reflective road marking}
\acrodef{rcs}[RCS]{radar cross section}
\acrodef{ris}[RIS]{reflective intelligent surface} 
\acrodef{rmse}[RMSE]{root mean square error}
\acrodef{rtt}[RTT]{round-trip time}
\acrodef{rsu}[RSU]{roadside unit}
\acrodef{simo}[SIMO]{single-input multiple-output}
\acrodef{siso}[SISO]{single-input single-output}
\acrodef{slam}[SLAM]{simultaneous localization and mapping}
\acrodef{snr}[SNR]{signal-to-noise ratio}
\acrodef{sota}[SoTA]{state-of-the-art}
\acrodef{sp}[SP]{scatter point}
\acrodef{srnf}[SR-NF]{single-reflector multipath \ac{nf}}
\acrodef{srs}[SRS]{sounding reference signal}
\acrodef{svd}[SVD]{singular value decomposition}
\acrodef{tdoa}[TDoA]{time difference of arrival}
\acrodef{toa}[ToA]{time of arrival}
\acrodef{tls}[TLS]{total least squares}
\acrodef{qcqp}[QCQP]{quadratically constrained quadratic program}

\acrodef{ue}[UE]{user equipment}
\acrodef{ul}[UL]{uplink}
\acrodef{ula}[ULA]{uniform linear array}
\acrodef{uwb}[UWB]{ultra-wideband}
\acrodef{ura}[URA]{uniform rectangular array}
\acrodef{v2x}[V2X]{vehicle-to-everything}
\acrodef{va}[VA]{virtual anchor}
\acrodef{vue}[VUE]{virtual \ac{ue}}
\acrodef{xl}[XL]{extremely large}

\articletype{Research Article}%

\received{Date Month Year}
\revised{Date Month Year}
\accepted{Date Month Year}
\journal{Journal}
\volume{00}
\copyyear{2026}
\startpage{1}

\newcommand{\PaperTitle}{Towards 6G Single-Anchor Vehicle Localization Exploiting Radio-Reflective Road Markings in Tunnel Environments}

\begin{document}

%\title{In-Tunnel Single-Anchor Localization Exploiting Near-Field and Radio-Reflective Road Markings}
\title{\PaperTitle}

\author[1]{Lorenzo Italiano}

\author[1]{Mattia Brambilla}

\author[2]{Monica Nicoli}

%\authormark{Italiano \textsc{et al.}}
\authormark{L. Italiano, M. Brambilla, M. Nicoli }
%\titlemark{In-Tunnel Single-Anchor Localization Exploiting Near-Field and Radio-Reflective Road Markings}
\titlemark{\PaperTitle}

\address[1]{\orgdiv{Dipartimento di Elettronica, Informazione e Bioingegneria}, \orgname{Politecnico di Milano}, \orgaddress{\state{Milan}, \country{Italy}}}

\address[2]{\orgdiv{Dipartimento di Ingegneria Gestionale}, \orgname{Politecnico di Milano}, \orgaddress{\state{Milan}, \country{Italy}}}

\corres{Corresponding author: Lorenzo Italiano \email{lorenzo.italiano@polimi.it}}

% \presentaddress{This is sample for present address text this is sample for present address text.}

%\fundingInfo{Text}
%\JELinfo{ejlje}

\abstract[Abstract]{Accurate vehicular localization remains a key challenge for cooperative intelligent transport systems (C-ITS), especially in areas without global navigation satellite system (GNSS) coverage, such as road tunnels. 
This paper proposes a novel vehicle positioning method with a single anchor equipped with multiple antennas, exploiting near-field (NF) propagation and passive radio-reflective structures deployed along the GNSS-denied tunnel. 
The method assumes a wideband vehicle-to-everything (V2X) communication between the vehicle and the anchor, in line with the undergoing standardization of cellular V2X beyond 5G.
We first derive the validity condition that allows us to approximate the multipath channel with a single reflector point,  defining a geometry validity bound on the number of antennas that can be employed. 
Building on this result, we propose JAVELIN, a 6G-compatible single-anchor localization framework that leverages tensor-based NF parameter estimation, adaptive NF/far-field (FF) processing, and recursive Bayesian tracking to enable sub-meter positioning without multi-anchor synchronization.
The method integrates angle, delay difference, and curvature measurements into a variable-dimension extended Kalman filter with gated nearest-neighbor  association, enabling operation without prior environmental knowledge. Radio-reflective road markings  are further introduced to enhance geometric diversity. Simulation results in realistic tunnel scenarios demonstrate accurate and robust localization under different conditions, outperforming state-of-the-art single-anchor approaches and benefiting from passive reflector deployment.}


% \jnlcitation{\cname{%
% \author{Taylor M.},
% \author{Lauritzen P},
% \author{Erath C}, and
% \author{Mittal R}}.
% \ctitle{On simplifying ‘incremental remap’-based transport schemes.} \cjournal{\it J Comput Phys.} \cvol{2021;00(00):1--18}.}

\maketitle

% \renewcommand\thefootnote{}
% \footnotetext{\textbf{Abbreviations:} ANA, anti-nuclear antibodies; APC, antigen-presenting cells; IRF, interferon regulatory factor.}

\renewcommand\thefootnote{\fnsymbol{footnote}}
\setcounter{footnote}{1}

\section{Introduction}\label{sec1}

\Ac{hap} is a key enabling technology for \acp{c-its} to support advanced services for  \acp{cav}~\cite{5gaa2021hap}.
While \ac{gnss} remains the de facto solution for open-sky scenarios,
its performance degrades severely in signal-blocked environments such as urban canyons, indoor facilities, and, in particular, road tunnels. In these scenarios, the absence of \ac{los}
satellite signals lead to large positioning errors or complete service outages, creating the need for complementary infrastructure-based localization solutions~\cite{saleh2026vehicularwirelesspositioning}.
This need is further highlighted by recent industry reports advocating the integration of heterogeneous sensing and positioning technologies, as well as network-assisted positioning services, to improve accuracy and situational awareness in connected and automated mobility~\cite{5gaa_positioning_paas}.

Several technologies have been proposed to address \ac{gnss}-denied vehicular positioning.
\Ac{uwb} systems provide high accuracy in short-range deployments~\cite{Piavanini2026roadsidepositioning, wen_automated_2020}, but they require dedicated and densely deployed infrastructure.
\Ac{ble} offers a lower-cost alternative, achieving meter-level accuracy when combined with novel features and advanced processing approaches~\cite{BLEtunnel, santra2024enhancing}.
LiDAR-based localization exploits recognized landmarks and lane markings with the support of digital maps~\cite{KimTunnel22}, while cooperative positioning has also been investigated as a means to improve accuracy~\cite{barbieri2024deep}.
In parallel, the recently standardized \ac{5g} advanced cellular technology, along with the ongoing evolution towards \ac{6g}, is expected to enable sub-meter positioning as a service by leveraging existing smart roads with \ac{v2x} infrastructure augmented by wideband and multi-antenna capabilities~\cite{italiano2025tutorial, SA-JURE}.
Currently standardized \ac{5g} positioning methods rely on geometric approaches for localization using time and angle, or measurements from multiple \acp{bs}, often under the assumption of perfect synchronized and \ac{los} propagation. However, in long tunnels, these assumptions may not hold, while the deployment of multiple anchors may be either costly or impractical. In the \ac{6g} vision, large antenna array deployments with \ac{nf} effects enable richer spatial information that can be exploited even in single-anchor configurations for accurate positioning. \Ac{nf} propagation modeling is more advanced than the conventional ray-based \ac{ff} model as it accounts for the spherical propagation wavefront rather than employing the conventional planar approximation. It provides a more accurate representation of radio propagation in short-range links with large antenna arrays, as envisioned for \ac{6g}-connected vehicle scenarios.

Exploiting \ac{nf} propagation with large antenna arrays provides richer geometric information on user location compared to classical \ac{ff} \ac{aoa} measurements~\cite{wang_near-field_2025, ebadi_near-field_2025}.
In fact, the \ac{nf} wavefront curvature enables the estimation of the transmit/receive source in three dimensions, potentially allowing single-anchor localization even in \ac{nlos} conditions.
This is not viable with conventional localization methods that require multiple anchors. Furthermore, single-anchor localization has the advantage of not requiring synchronization among anchors, thus reducing the complexity and cost of the infrastructure and intrinsically solving a practical impairment that affects current networks \cite{BarbieriTIM25}.

Multipath is widely recognized as a major localization impairment in indoor environments, such as tunnels. Recent studies on \ac{nf} localization in multipath environments have been presented in the literature~\cite{yang2025near}. A common modeling approach is the \ac{srnf} channel model, where the dominant \ac{nlos} paths between the \ac{ue} and all the elements of the antenna array at \ac{bs} are assumed to originate from a unique reflection point on a planar surface~\cite{TeNFiLoc, 10620236, 10955747}. Under this assumption, the reflector can be interpreted as the effective wave origin, and its position can be inferred from the spatial phase profile across the array. However, this model implicitly assumes that the reflection point is shared across all antenna elements. In general, this is not strictly true: the physically correct representation corresponds to a \ac{va}, which in the considered \ac{ul} scenario coincides with a \ac{vue}, whose location depends on specular geometry. The conditions under which the \ac{srnf} interpretation remains geometrically consistent have not been formally characterized.

\begin{figure}[t]
    \centering
    \begin{tikzpicture}
        \node[]at(0,0){\includegraphics[width=0.6\linewidth]{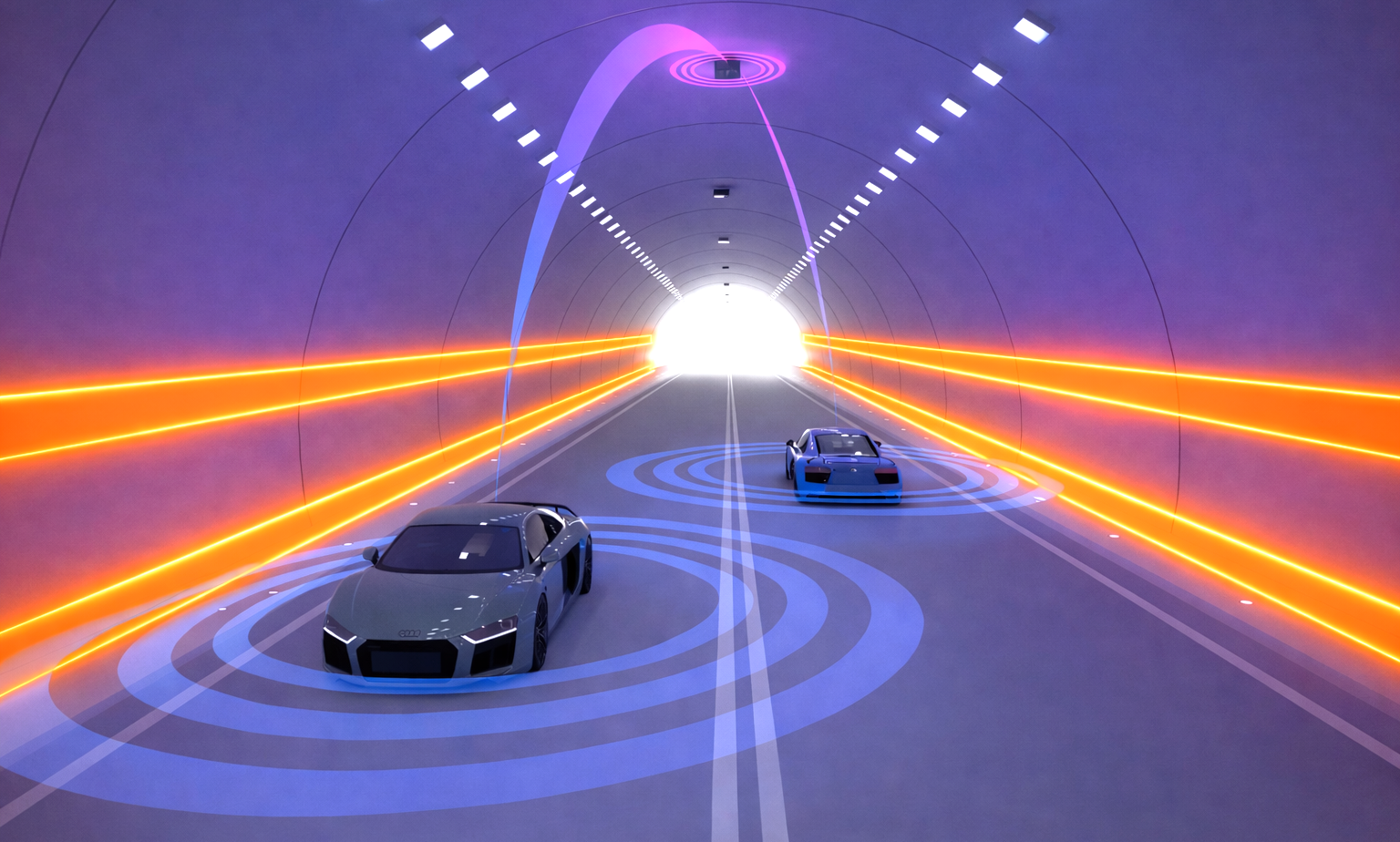}};
        %\node[white]at(0.2,2.2){BS/RSU};
    \end{tikzpicture}
    \caption{Tunnel localization scenario with two CAVs exploiting a single anchor for positioning. The orange panels are \acp{rrm} used to improve the vehicle localization accuracy.}
    \label{fig:tunnel_scenario}
\end{figure}

In this paper, we focus on vehicle localization in tunnel environments, where the vehicle is connected to a single anchor (which is indistinctly identified by a cellular \ac{bs} or \ac{rsu}  in our formulation) 
%cellular \ac{rsu} (here referred to as \ac{bs})
and position-related measurements are obtained from the radio link, accounting for both \ac{los} and reflected \ac{nlos} propagation paths. Reflections are controlled through the deployment of \acp{rrm}, or passive reflectors. The considered scenario, including an example of passive reflector deployment, is illustrated in Figure~\ref{fig:tunnel_scenario}. The contribution is twofold. First, we derive a geometric validity condition for the \ac{srnf} channel model in tunnel scenarios, addressing the lack of consistency conditions between physical reflectors and channel representations. We show that, for a given propagation distance and wavelength range, there exists a maximum array size for which the channel can be equivalently interpreted as generated by a single physical reflector (see Figure~\ref{fig:topview}). Beyond this threshold, the channel can only be represented via \acp{va}. Remarkably, the resulting bound follows the classical \ac{nf} scaling, revealing a direct connection between geometric consistency and Fresnel-region behavior. 
\begin{figure}
    \centering
    \include{figures/top_view_scenario}
    \vspace{-1cm}
    \caption{Top-view of the UE-BS communication scenario with single bounce reflection. The radio signal propagates from the UE to the multiple antennas at the BS through the reflected paths. The solid blue lines represent the true paths connecting the \ac{ue} to the \ac{bs} antenna elements; the dashed blue lines represent the corresponding paths originated from the specular image of the UE, i.e., the VUE; and the dotted red lines represent the paths assumed in the SR-\ac{nf} channel model, where all reflection points on the tunnel wall coincide at $\mathbf p_{s,0}$. }
    \label{fig:topview}
\end{figure}
As a second contribution, building on this theoretical foundation, we propose JAVELIN (\textit{Joint and Adaptive Virtual and Ego-user Localization In Near-field}), a \ac{6g}-oriented single-anchor localization framework tailored to vehicular tunnel environments. The method leverages \ac{nf} parameter extraction via tensor-based channel decomposition, combines the \ac{nf} and \ac{ff} regimes adaptively, and integrates measurements into an adaptive, variable-state \ac{ekf} using a gated \ac{nn}-based association strategy. This work extends our former proposal~\cite{SA-JURE} by exploiting \ac{nf} information to remove the need for digital maps and prior knowledge of reflector locations. 
Furthermore, we discuss the deployment of \acp{rrm}, enabling a low-cost, infrastructure-light localization paradigm. Unlike conventional multi-anchor solutions, \acp{rrm} provide passive geometric anchors that enhance positioning robustness while maintaining scalability, aligning with emerging smart-road deployment strategies.
Conceptually, these reflectors play a role analogous to road lane markings for human drivers: just as visual references guide the vehicle trajectory in low-visibility conditions, engineered reflectors provide geometric anchors that guide the positioning system in \ac{gnss}-denied environments~\cite{kakkavas2021position}.

The remainder of the paper is organized as follows.
Section~\ref{sec:system_model} introduces the system and channel models.
Section~\ref{sec:methodology} presents the proposed localization methodology and the \ac{srnf} validity theorems.
Section~\ref{sec:results} evaluates the performance in a realistic tunnel simulation environment.
Section~\ref{sec:conclusions} concludes the paper.

\subsection{Related Works}
In the following, we review the state-of-the-art methodologies for infrastructure-based localization suited for tunnel environments, as well as single-anchor positioning architectures.

Recent works investigating vehicle localization in tunnels have proposed exploiting \ac{tdoa} measurements obtained from commercial \ac{uwb} roadside infrastructure. Specifically, the approach in~\cite{wen_automated_2020} proposes a \ac{tdoa}-based architecture enabling real-time localization with multiple anchors and edge computation, achieving sub-meter accuracy even at relatively high vehicle speeds. However, the approach relies on \ac{los} geometric multilateration solved via \ac{ls}, which does not explicitly account for vehicle dynamics or measurement nonlinearities. Conversely, the work in \cite{PiavaniniTunnel25} introduces a more advanced probabilistic framework based on a nonlinear variational Bayes multiple model (N-VBMM), which jointly handles nonlinear measurement models and multiple motion hypotheses. This results in improved tracking robustness and positioning accuracy, particularly in complex driving scenarios such as lane changes or slalom maneuvers. Yet it requires a dense, perfectly synchronized deployment. The single-anchor alternative is infrastructure-efficient and does not require synchronization.

A first work investigating single-anchor positioning in tunnels is~\cite{halili2022vehicle}, which exploits \ac{v2x} communications to enable continuous localization by combining onboard sensor information with Doppler and \ac{toa} measurements at the roadside infrastructure. However, the framework primarily relies on direct-path components, neglecting multipath and \ac{los} obstruction, which may significantly degrade the performance in real environments.
To exploit multipath, the authors in~\cite{SA-JURE} propose a single-anchor vehicle localization methodology based on \acp{aoa} and single-anchor \acp{tdoa}, leveraging reflections from tunnel walls together with prior knowledge of the environment geometry. Reflectors are modeled as virtual anchors, and the vehicle and reflector states are jointly estimated through an \ac{ekf}, achieving robust performance even in \ac{nlos} conditions and without strict synchronization requirements. Nevertheless, the evaluation is conducted under simplified simulation settings (e.g., limited trajectories and regular tunnel geometries), and the method assumes accurate prior knowledge of the environment, which may limit its applicability in more complex or irregular scenarios.
Similarly, \cite{PengTunnelSBLOC25} proposes a multipath-assisted localization strategy in a single-anchor setup, introducing a dedicated multipath selection algorithm tailored for tunnel environments. The method filters out higher-order and non-wall reflections (e.g., clutter or ceiling reflections) using geometric constraints and estimated channel parameters (e.g., \ac{aoa}, \ac{toa}), thereby improving localization accuracy. Simulation results show a significant reduction in positioning error when the selection algorithm is applied. However, the approach assumes simplified tunnel geometries and relies on accurate multipath parameter estimation, which may be challenging in real-world deployments.

Extending the review of single-anchor positioning outside the tunnel context, the authors in~\cite{sun2021comparative} compare different \ac{ul} \ac{5g} positioning algorithms without explicitly addressing multipath exploitation or temporal dynamics. In~\cite{shamsian2023joint}, a joint single-anchor \ac{tdoa} and \ac{aoa} framework is proposed to estimate both the \ac{ue} and reflector positions through a two-step procedure; however, the separation between reflector and \ac{ue} estimation may introduce suboptimality and increased latency in dynamic scenarios. A real-world experimental validation is presented in~\cite{ge2024experimental}, where \ac{rtt}, \ac{aoa}, and \ac{aod} measurements are leveraged to perform \ac{slam} in a vehicular scenario; nevertheless, \ac{rtt}-based techniques typically require multiple message exchanges, which may limit their applicability in highly dynamic environments due to increased latency. In~\cite{bader2024leveraging}, the authors exploit single-bounce reflections to jointly estimate \ac{aoa} and \ac{aod} using large antenna arrays at both transmitter and receiver sides; while effective, this assumption may be impractical in many real-world deployments due to hardware and cost constraints.
Finally, \cite{nazari2023mmwave} proposes a high-resolution mmWave localization framework capable of estimating the full 6D user state from a single anchor using a snapshot of channel parameters. While the method achieves high accuracy by leveraging angular and delay information, it relies on a static snapshot model and does not incorporate temporal tracking or data association mechanisms, thus limiting its robustness in dynamic scenarios and in the presence of measurement uncertainty and missed detections.

Overall, the reviewed literature highlights three main trends: \textit{(i)} multi-anchor approaches achieving high accuracy through geometric or probabilistic inference, often at the cost of infrastructure complexity; \textit{(ii)} single-anchor tunnel-specific methods that exploit multipath as virtual anchors, but typically rely on simplified geometries or prior environmental knowledge; and \textit{(iii)} general single-anchor frameworks that either neglect multipath, adopt suboptimal multi-stage estimation strategies, or rely on impractical assumptions such as large antenna arrays or high-latency measurements. Moreover, snapshot-based solutions lack temporal tracking and robustness in dynamic scenarios, and typically do not exploit the additional information provided by \ac{nn} propagation. In the following sections, to address these limitations, we propose a unified framework that jointly estimates the \ac{ue} and environmental features over time, it explicitly exploits multipath and \ac{nn} propagation without requiring prior knowledge of the environment, and it also integrates adaptive data association and track management along with Bayesian filtering for dynamic management. This enables robust and scalable localization in complex dynamic tunnel scenarios using a single-anchor architecture.

\subsection{Contributions}
The contributions are summarized as follows:
\begin{itemize}

\item We propose the deployment of passive radio markers along roadways, enabling a scalable and infrastructure-light localization paradigm in which passive elements act as opportunistic virtual anchors, significantly enhancing the geometric diversity without requiring additional active infrastructure (i.e., multiple anchors).

\item We present the JAVELIN method, a robust single-anchor \ac{6g}-oriented positioning
framework tailored to tunnel environments, leveraging location-related parameter extraction and adaptive \ac{nf}/\ac{ff} processing without requiring digital map assistance or prior knowledge of reflector locations.

\item We derive the validity condition that allows to assume the \ac{srnf} channel model, establishing a constraint on the array size as a function of the propagation distance, wavelength, and specific environmental parameters. The result formally characterizes the transition between single-reflector and \ac{va}-based interpretations of \ac{nf} channels.

\item We integrate the extracted channel parameters into a recursive state-estimation architecture based on an adaptive \ac{ekf} with \ac{nn}-based data association, enabling seamless vehicular tracking in \ac{gnss}-denied conditions.

\item We carry out 
%extensive performance analyses 
numerical simulations
in a realistic tunnel environment, validating both the theoretical findings and the practical feasibility of the proposed approach.

\end{itemize}

\paragraph*{Notation}
Matrices are defined with bold uppercase (e.g., $\mathbf{A}$), vectors with bold lowercase (e.g., $\mathbf{a}$), and tensors with calligraphic (e.g., $\mathcal{A}$). Element indices are indicated with lowercase subscripts (e.g., $\mathbf{a}_i^{}$). The operations include transpose $(\cdot)^\mathsf{T}$, conjugate transpose $(\cdot)^\mathsf{H}$, and mode-$i$ tensor product ($\times_i$) between the $i$-th dimension of the first tensor and the \nth{2} dimension of the second one. The diag$(\cdot)$ operator is used to create a diagonal matrix from a vector (e.g., diag$(\mathbf{a})$). 
% The operator $\mathcal{R}(\cdot)$ denotes the range of a matrix, and vec$(\cdot)$ denotes vectorization. 
Moreover, $\left[\cdot\right]_j$ denotes unfolding over the $j$-th mode, $\Vert \cdot\Vert$ denotes the $\ell_2$ norm, $\mathbf{I}_M\in\mathbb{C}^{M\times M}$ is the identity matrix, and $\mathbf{0}_{M\times N}\in\mathbb{C}^{M\times N}$ is the all-zero matrix. The projector onto the column space of $\mathbf{A}$ is $\mathbf{P}_{\mathbf{A}}=\mathbf{A}\mathbf{A}^\dagger$, and the corresponding orthogonal projector is $\mathbf{P}^{\perp}_{\mathbf{A}}=\mathbf{I}-\mathbf{P}_{\mathbf{A}}$. Finally, $\mathbf{A}_{+}$ and $\mathbf{A}_{-}$ denote the matrix $\mathbf{A}$ without the first and last row, respectively.

\section{System Model}
\label{sec:system_model}
We consider an \ac{ul} scenario where a single anchor (e.g., a \ac{bs} or \ac{rsu}) with $M$ antennas is located at a known position and orientation inside a tunnel. 
The system model aligns with \ac{6g} architectures characterized by large-aperture antenna arrays and \ac{nf} propagation conditions, where spherical wavefront effects must be explicitly accounted for.
The \ac{ue}, equipped with a single antenna, is not synchronized with the \ac{bs} and travels through the tunnel while broadcasting kinematic information (including position and velocity) via \ac{v2x} messages~\cite{cam}. The position and velocity data are obtained from the vehicle’s onboard sensors (e.g., \ac{gnss} and speedometer). Upon entering the tunnel, however, reliable position information from \ac{gnss} becomes unavailable, whereas other kinematic measurements, such as velocity, remain valid. Since the reduced reliability of \ac{gnss} prevents the provision of accurate positioning information, our goal is to develop a methodology that serves as an alternative for seamless, accurate positioning within tunnels.
To this end, we opportunistically exploit radio-reflective road markings and other static structures within the tunnel environment, acting as additional virtual anchors. By leveraging the multipath components generated by these reflectors, the system effectively increases the available spatial diversity, enabling improved localization accuracy even in a single-anchor setup.

\subsection{Channel and Signal Model}
\label{sec:chsigmod}
We consider a  \ac{simo}-\ac{ofdm} wireless  channel composed of $L$ paths between a \ac{ue} and the \ac{bs}. The baseband equivalent channel response at subcarrier $s$, antenna $m$, and symbol $k$ is modeled as
\begin{equation}
    h^{}_{m,s}[k] = \sum^{L-1}_{\ell=0}\alpha^{}_\ell\, e^{j2\pi \frac{f^{}_c}{c}\delta^{}_{\ell,m}} \, e^{-j2\pi s\Delta f\tau^{}_\ell} \, e^{j2\pi k f^{d}_\ell T^{}_0},
    \label{eq:channel_model}
\end{equation}
where $\alpha^{}_\ell$ is the complex channel gain, $f^{}_c$ is the carrier frequency, $c$ is the speed of light, $\Delta f$ is the subcarrier spacing, $\tau^{}_\ell$ is the path delay, $f^{d}_\ell$ is the Doppler shift, and $T^{}_0$ is the sampling interval. The term $\delta^{}_{\ell,m}$ denotes the propagation distance offset for path $\ell$ at antenna element $m$ to the reference antenna (with $m=0$) and depends on the adopted wavefront model. Specifically, we adopt a spherical wavefront model and define
\begin{equation}
    \delta^{}_{\ell,m} \triangleq d^{}_{\ell,m} - d^{}_{\ell,0},
    \label{eq:delta_model}
\end{equation}
where $d^{}_{\ell,m}$ is the geometric path length from the transmitter to the $m$-th antenna along path $\ell$. In~\cite{TeNFiLoc}, the authors model $\delta^{}_{\ell,m} = \delta^{}_m(x^{}_\ell, y^{}_\ell, z^{}_\ell)$ as a function of the last reflection point. We refer to that channel model as the \ac{srnf} channel.

It is convenient to represent the channel in tensor form $\mathcal{H} \in \mathbb{C}^{M\times N^{}_f\times N^{}_t}$ admitting the Tucker, or \ac{cp}, decomposition as 
\begin{equation}
    \mathcal{H} = \mathcal{A}\times_1 \mathbf{B}^s\times_2 \mathbf{B}^f\times_3 \mathbf{B}^t,
    \label{eq:channel_tensor_model}
\end{equation}
being $\mathcal{A}$ diagonal with entries $\mathcal{A}^{}_{(\ell,\ell,\ell)}=\alpha_\ell$, \(\mathbf{B}^s = \begin{bmatrix}
    \mathbf{b}^{s}_0 \, \cdots \, \mathbf{b}^{s}_{L-1}
\end{bmatrix} \in \mathbb{C}^{M\times L}\) the spatial-domain steering matrix, \(\mathbf{B}^f = \begin{bmatrix}
    \mathbf{b}^{f}_0 \, \cdots \, \mathbf{b}^{f}_{L-1}
\end{bmatrix} \in \mathbb{C}^{N^{}_f\times L}\) the frequency-domain steering matrix, and 
\(\mathbf{B}^t= \begin{bmatrix}
    \mathbf{b}^{t}_0 \, \cdots \, \mathbf{b}^{t}_{L-1}
\end{bmatrix} \in \mathbb{C}^{N^{}_t\times L}\) the time-domain steering matrix, with
\begin{align}
    \mathbf{b}^{s}_\ell &= \begin{bmatrix}1 & e^{j2\pi \frac{f^{}_c}{c}\delta^{}_{\ell,1}} & \cdots & e^{j2\pi \frac{f^{}_c}{c}\delta^{}_{\ell,M-1}}\end{bmatrix}^\mathsf{T}, 
    \label{eq:bs}
    \\
    \mathbf{b}^{f}_\ell &= \begin{bmatrix}1 & e^{-j2\pi \Delta f\tau^{}_\ell} & \cdots & e^{-j2\pi \left(N^{}_f-1\right)\Delta f\tau^{}_\ell}\end{bmatrix}^\mathsf{T}, 
    \label{eq:bf}
    \\
    \mathbf{b}^{t}_\ell &= \begin{bmatrix}1 & e^{j2\pi T^{}_0 f^{d}_\ell} & \cdots & e^{-j2\pi \left(N^{}_t-1\right) T^{}_0 f^{d}_\ell}\end{bmatrix}^\mathsf{T}.
    \label{eq:bt}
\end{align}
The resulting received signal model in tensor form is:
\begin{equation}
\mathcal{Y} = \mathcal{H} \odot \mathcal{X} + \mathcal{Z},
\label{eq:rx_signal_tensor_model}
\end{equation}
where $\mathcal{Y}$, $\mathcal{X}$, and $\mathcal{Z}\in\mathbb{C}^{M\times N^{}_f\times N^{}_t}$ are third-order tensors indexed by $(m,s,k)$, $\mathcal{X}$ denotes the transmitted signal, and $\mathcal{Z}^{}_{(m,s,k)}\sim\mathcal{CN}(0,N^{}_0)$ denotes the noise, with $N^{}_0$ the noise power spectral density.
Assuming the transmitted signal is fully known at the receiver, the transmitted data can be removed from the received samples by element-wise division, resulting in
\begin{equation}
    \tilde{\mathcal{H}} = \mathcal{A}\times_1 \mathbf{B}^s\times_2 \mathbf{B}^f\times_3 \mathbf{B}^t + \tilde{\mathcal{Z}},
    \label{eq:channel_tensor_approx}
\end{equation}
where
\begin{align}
    \tilde{\mathcal{H}}^{}_{(m,s,k)} &= \mathcal{Y}^{}_{(m,s,k)}/\mathcal{X}^{}_{(m,s,k)},\\
    \tilde{\mathcal{Z}}^{}_{(m,s,k)} &= \mathcal{Z}^{}_{(m,s,k)}/\mathcal{X}^{}_{(m,s,k)}.
    \label{eq:channel_estimation}
\end{align}

\subsection{Location Parameter Estimation}
According to the modeling in Section~\ref{sec:chsigmod}, the $\ell$-th path is characterized by the unknown parameter vector \(\bm{\rho}^{}_\ell = \begin{bmatrix}
    \alpha^{}_\ell & \phi^{}_\ell & \psi^{}_\ell & \kappa^{}_\ell & d^{}_\ell & v^{}_\ell
\end{bmatrix}^\mathsf{T}\), where $\phi^{}_\ell$ and $\psi^{}_\ell$ are the azimuth and elevation \acp{aoa} of the $\ell$-th path, $\kappa^{}_\ell$ is the wavefront radius of curvature (which captures the geometric distance from the $\ell$-th wave origin), $d^{}_\ell = \tau^{}_\ell/c$ is the propagation distance of $\ell$-th path affected by the clock bias, and $v^{}_\ell = c\cdot f^{d}_\ell/f^{}_c$ is the relative velocity along the $\ell$-th path. 

For measuring the location parameters needed for localization, we consider the TeNFiLoc algorithm~\cite{TeNFiLoc}, using the \ac{cp} decomposition to estimate the steering matrices $\mathbf{B}^s$, $\mathbf{B}^f$, and $\mathbf{B}^t$, which must satisfy the \textit{Kruskal's uniqueness condition}. The \ac{cp} decomposition is essentially unique if $\text{kr}(\mathbf{B}^s)+\text{kr}(\mathbf{B}^f)+\text{kr}(\mathbf{B}^t) \ge 2L+2$, where $\text{kr}(\cdot)$ denotes the Kruskal rank. Let $\widehat{\mathbf{B}}^t$, $\widehat{\mathbf{B}}^f$, and $\widehat{\mathbf{B}}^t$ denote the steering matrix estimates obtained as
\begin{equation}
    \{\widehat{\mathbf{B}}^s,\,\widehat{\mathbf{B}}^f,\,\widehat{\mathbf{B}}^t\} \leftarrow \text{CPD}(\tilde{\mathcal{H}}, L),
\end{equation}
with $\text{CPD}(\cdot)$ the \ac{cp} decomposition function.
Since all steering matrices have unitary first row (see \eqref{eq:bs}, \eqref{eq:bf}, \eqref{eq:bt}), we can estimate the path gain as
\begin{equation}
    \widehat{\alpha}^{}_\ell = \widehat{\mathbf{B}}^s_{(0,\ell)} \cdot \widehat{\mathbf{B}}^f_{(0,\ell)} \cdot \widehat{\mathbf{B}}^t_{(0,\ell)}, \quad \forall \, 0 \leq \ell < L.
\end{equation}
Given a generic steering matrix $\widehat{\mathbf{B}}$, the scaling ambiguity can be resolved as \(\widehat{\mathbf{B}} \leftarrow \widehat{\mathbf{B}} \cdot \text{diag}^{-1}\left(\widehat{\mathbf{B}}^{}_{(0,:)}\right)\).
Each matrix is then Vandermonde, and we compute the roots by exploiting the shift-invariance property as
\begin{equation}
    \widehat{\bm \mu} = \angle \left(\text{diag}\left(\widehat{\mathbf{B}}^{\mathsf{H}}_{-} \widehat{\mathbf{B}}_{+}\right)\right).
\end{equation}
Thereby, $\widehat{\mathbf d}=[\widehat{d}^{}_0 \cdots \widehat{d}^{}_{L-1}]$ and $\widehat{\mathbf v}=[\widehat{v}^{}_0 \cdots \widehat{v}^{}_{L-1}]$ can be estimated as follows
\begin{align}
    \widehat{\mathbf d} &= - \frac{c}{2\pi \Delta f} \angle \left(\text{diag}\left((\widehat{\mathbf{B}}^f_{-})^\mathsf{H}\widehat{\mathbf{B}}^f_{+}\right)\right),\\
    \widehat{\mathbf v} &= \frac{c}{2\pi f^{}_c T^{}_0} \angle \left(\text{diag}\left((\widehat{\mathbf{B}}^t_{-})^\mathsf{H}\widehat{\mathbf{B}}^t_{+}\right)\right).
\end{align}
The estimation of $\phi^{}_\ell$, $\psi^{}_\ell$, and $\kappa^{}_\ell$ consists of three steps, summarized below.
\subsubsection{Phase Unwrapping}
The estimate of the path difference is obtained by unwrapping the phase of the steering vector $\widehat{\mathbf{b}}^s_\ell$ as
\begin{equation}
    \widehat{\bm\delta}^{}_\ell = \frac{c}{2\pi f^{}_c}\mathcal{U}\left(\angle\widehat{\mathbf{b}}^s_\ell\right),
\end{equation}
where $\mathcal{U}(\cdot)$ denotes the 2D phase-unwrapping function~\cite{Herraez2002fast2Dphaseunwrap}.

\subsubsection{Linear System Solution}
Let $\mathbf{p}^{}_\ell = (x^{}_\ell, y^{}_\ell, z^{}_\ell)$ be the virtual wave origin (which corresponds to the last reflector in the \ac{srnf} model), and the reference antenna placed in the origin. Then, it is:
\begin{equation}
    \widehat{\delta}^{}_{\ell,m}+\kappa^{}_\ell = \sqrt{(x^{}_m-x^{}_\ell)^2+(y^{}_m-y^{}_\ell)^2+(z^{}_m-z^{}_\ell)^2}.
\end{equation}
Following the expansion and the definition of the linear system in~\cite{TeNFiLoc}, and using the \ac{ura} version to obtain a full-column-rank matrix, we have
\begin{equation}
    2\underbrace{\begin{bmatrix}
    x^{}_1 & y^{}_1 & \widehat{\delta}^{}_{\ell,1} \\
    x^{}_2 & y^{}_2 & \widehat{\delta}^{}_{\ell,2} \\
    \vdots & \vdots & \vdots \\
    x^{}_{M-1} & y^{}_{M-1} & \widehat{\delta}^{}_{\ell,M-1}
    \end{bmatrix}}_{\mathbf{A}^{}_\ell \in \mathbb{R}^{(M-1)\times3}} \mathbf{x}^{}_\ell = \underbrace{\begin{bmatrix}
    r^2_1-\widehat{\delta}^{2}_{\ell,1} \\
    r^2_2-\widehat{\delta}^{2}_{\ell,2} \\
    \vdots \\
    r^2_{M-1}-\widehat{\delta}^{2}_{\ell,M-1}
    \end{bmatrix}}_{\mathbf{y}^{}_\ell \in \mathbb{R}^{(M-1)\times1}},
\end{equation}
where $r^2_m=x^2_m+y^2_m+z^2_m$ is the squared distance between the $m$-th antenna and the reference antenna. Moreover,
\begin{equation}
\mathbf{x}^{}_\ell = 
\begin{bmatrix}
x^{}_\ell \\
y^{}_\ell \\
\kappa^{}_\ell
\end{bmatrix} = 
\begin{bmatrix}
\kappa^{}_\ell \cos\phi^{}_\ell \sin\psi^{}_\ell\\
\kappa^{}_\ell \sin\phi^{}_\ell \sin\psi^{}_\ell\\
\kappa^{}_\ell
\end{bmatrix}.
\label{eq:x_l}
\end{equation}
Note that the antenna coordinates must be expressed in a local coordinate system in which the antenna array lies in the $xy$-plane, i.e., $z_m=0$ $\forall m$.
Since both $\mathbf{A}_\ell$ and $\mathbf{y}^{}_\ell$ include the noisy term $\widehat{\delta}_{\ell,m}$, we estimate $\mathbf{x}_\ell$ via \ac{tls} solution using the augmented matrix
    $\begin{bmatrix}
        2\mathbf{A}_\ell & -\mathbf{y}_\ell
    \end{bmatrix}$
for the \ac{svd}.
In the \ac{ff} regime, we define
$\mathbf{A}^\mathrm{FF}_\ell=\begin{bmatrix}\mathbf{A}_\ell\end{bmatrix}_{(:,1:2)}$ and
$\mathbf{x}^\mathrm{FF}_\ell=\begin{bmatrix}\cos\phi_\ell \sin\psi_\ell & \sin\phi_\ell \sin\psi_\ell\end{bmatrix}^\mathsf{T}$,
and solve via \ac{ls}.

\subsubsection{Parameter Extraction}
We estimate $\phi^{}_\ell$, $\psi^{}_\ell$, and $\kappa^{}_\ell$ as
\begin{align}
    \widehat{\kappa}^{}_\ell &= [\widehat{\mathbf{x}}^{}_\ell]_{(3)} ,\\
    \widehat{\phi}^{}_\ell &= \text{atan}\left([\widehat{\mathbf{x}}^{}_\ell]^{}_{(2)},[\widehat{\mathbf{x}}^{}_\ell]^{}_{(1)}\right) , \\
    \widehat{\psi}^{}_\ell &= \cos^{-1}\left(\frac{\sqrt{[\widehat{\mathbf{x}}^{}_\ell]^2_{(1)}+[\widehat{\mathbf{x}}^{}_\ell]^2_{(2)}}}{\kappa^{}_\ell}\right),
\end{align}
or, in \ac{ff}, as
\begin{align}
    \widehat{\phi}^\text{FF}_\ell &= \text{atan}\left([\widehat{\mathbf{x}}^\text{FF}_\ell]^{}_{(2)},[\widehat{\mathbf{x}}^\text{FF}_\ell]^{}_{(1)}\right) ,\\
    \widehat{\psi}^\text{FF}_\ell &= \cos^{-1}\left(\sqrt{[\widehat{\mathbf{x}}^\text{FF}_\ell]^2_{(1)}+[\widehat{\mathbf{x}}^\text{FF}_\ell]^2_{(2)}}\right).
\end{align}
These parameters will be used in the next section for localization.

% \begin{center}
% \begin{table*}[!h]%
% \caption{This is sample table caption.\label{tab1}}
% \begin{tabular*}{\textwidth}{@{\extracolsep\fill}lllll@{}}
% \toprule
% &\multicolumn{2}{@{}l}{\textbf{Spanned heading$^{\tnote{\bf a}}$}} & \multicolumn{2}{@{}l}{\textbf{Spanned heading$^{\tnote{\bf b}}$}} \\\cmidrule{2-3}\cmidrule{4-5}
% \textbf{Col1 head} & \textbf{Col2 head}  & \textbf{Col3 head}  & {\textbf{Col4 head}}  & \textbf{Col5 head}   \\
% \midrule
% col1 text & col2 text  & col3 text  & 12.34  & col5 text\tnote{1}   \\
% col1 text & col2 text  & col3 text  & \hphantom{0}1.62  & col5 text\tnote{2}   \\
% col1 text & col2 text  & col3 text  & 51.809  & col5 text   \\
% \bottomrule
% \end{tabular*}
% \begin{tablenotes}%%[341pt]
% \item[$^{\rm a}$] Example for a first table footnote.
% \item[$^{\rm b}$] Example for a second table footnote.
% \item {\it Source}: Example for table source text.
% \end{tablenotes}
% \end{table*}
% \end{center}

\section{Localization Methodology}
\label{sec:methodology}

\subsection{Scenario and Assumptions}

We consider a \ac{bs} located at the origin, at the midpoint of a straight tunnel segment. It is equipped with two antenna arrays aligned with the tunnel longitudinal axis and pointing in opposite driving directions. The tunnel is assumed to have a semicircular cross-section that remains constant along the longitudinal axis. Accordingly, the tunnel cross-sectional plane (i.e., the plane containing the semicircle) is orthogonal to the array boresight direction.
Moreover, we assume that the number of paths, $L$, is known at the \ac{bs}, and that the vehicle velocity, estimated from onboard sensors, is shared via \ac{v2x}.
We adopt the \ac{srnf} channel model proposed in~\cite{TeNFiLoc}, while enforcing the array-aperture constraints required for the validity of the spherical \ac{nf} formulation. In particular, we impose an upper bound on the array size to ensure that the \ac{srnf} approximation remains accurate; the corresponding theoretical conditions are provided below.
Figure~\ref{fig:topview} shows a top view of the considered scenario and highlights the geometric setup used in the following theorem.

\begin{theorem}[\textbf{Validity condition for a 2D \ac{srnf} channel model}]
Consider a \ac{bs} equipped with a \ac{ula} of $M$ antennas with inter-element spacing $d=\lambda/2$, operating at wavelength $\lambda$, and a \ac{ue} located at a longitudinal (i.e., along the road x-axis) distance $R$ from the \ac{bs}.
Assume that the dominant \ac{nlos} path from the \ac{ue} to the \ac{bs} is generated by a single specular reflection on a planar surface perpendicular to the \ac{bs}, such that the reflection point is uniquely determined by the geometry of the \ac{bs}, the \ac{ue}, and the reflecting plane. Since the tunnel walls are perpendicular to the \ac{bs} orientation, the $x$-coordinates of the \ac{ue} and the \ac{vue} locations coincide.

Let $\mathcal{H}$ denote the \ac{nf} \ac{ul} channel observed across the \ac{bs} array, and let $\varepsilon^{}_\Phi$ be the maximum tolerable phase error.
Then, there exists a maximum number of \ac{bs} antennas $M_{\max}$ such that:

\begin{itemize}
    \item for $M \le M_{\max}$, the channel $\mathcal{H}$ can be equivalently represented by a single reflector located on the planar surface;
    \item for $M > M_{\max}$, the channel $\mathcal{H}$ can only be represented by distinct paths converging in a \ac{va} (e.g., a \ac{vue}), single-reflector approximation no longer holds.
\end{itemize}

The threshold $M_{\max}$ is upper bounded by
\begin{equation}
\label{eq:Nmax}
M_{\max}
\;\le\;
1 + 2\sqrt{
\frac{R\,W\,\varepsilon^{}_\Phi}{\lambda\,\left|W-y_u\right|\,\pi}
},
\end{equation}
where $W$ denotes the distance between the \ac{bs} and the reflecting surface and $y_u$ is the transverse coordinate of the \ac{ue}.
\end{theorem}

\begin{proof}
Using the \ac{vue} construction, the \ac{nlos} path via the planar reflector at $y=W$ is equivalent to a \ac{los} path from the \ac{bs} to the \ac{vue} located at
\[
\mathbf p_v^{} = (R,\,2W-y_u).
\]
Accordingly, the exact \ac{nlos} path lengths to the reference antenna and to the $m$-th antenna are given by
\begin{align}
d^{}_{v,0} &= \sqrt{R^2 + (2W-y_u)^2},
\\
d^{}_{v,m} &= \sqrt{R^2 + (2W-y_u-y_m)^2},
\end{align}
where $y_m = m d$, with $d$ the inter-element spacing.

%$|y_m|\ll R$
For $\left|y_m\right|\ll R$, a second-order Fresnel expansion yields
\begin{equation}
\label{eq:delta_l_exact}
\delta^\ast_m=d^{}_{v,m} - d^{}_{v,0}
\approx
\frac{y_m^2 - 2(2W-y_u)y_m}{2R}.
\end{equation}

Let $\mathbf p_{s,0}=(x_{p_0}, W)$ denote the specular reflection point associated with the reference antenna.
The corresponding distances to the \ac{bs} antennas are
\begin{align}
d_{p_0,0}&=\sqrt{x_{p_0}^2+W^2},
\\
d_{p_0,m}&=\sqrt{x_{p_0}^2+(W-y_m)^2}.
\end{align}
Applying the same Fresnel approximation gives
\begin{equation}
\label{eq:delta_l_sp}
\delta^\text{SR}_m=d_{p_0,m}-d_{p_0,0}
\approx
\frac{y_m^2 - 2Wy_m}{2x_{p_0}}.
\end{equation}

For a planar reflector, the specular point is uniquely determined by geometry.
Using the \ac{vue} method, its horizontal coordinate is
\begin{equation}
\label{eq:xp_0}
x_{p_0} = R\,\frac{W}{2W-y_u}.
\end{equation}
Substituting \eqref{eq:xp_0} into \eqref{eq:delta_l_sp} and comparing with
\eqref{eq:delta_l_exact}, the linear terms in $y_m$ cancel exactly.
The resulting approximation error is therefore
\begin{equation}
\label{eq:error}
\Delta\delta_m
\triangleq
\delta^\ast_m-\delta^\text{SR}_m
\approx
-\frac{W-y_u}{2RW}\,y_m^2.
\end{equation}
The maximum phase error occurs at the array edge $y_{\max}=(M-1)\,d$.
Imposing the phase error constraint
\[
\frac{2\pi}{\lambda}\,\left|\Delta\delta_m\right| \le \varepsilon^{}_\Phi
\]
and using \eqref{eq:error} yields
\begin{equation}
\frac{2\pi}{\lambda}
\frac{\left|W-y_u\right|}{2RW}
(M-1)^2 d^2
\le \varepsilon^{}_\Phi.
\end{equation}
With $d=\lambda/2$, solving for $M$ gives \eqref{eq:Nmax}, which completes the proof.
\end{proof}

\begin{remark}[(Scaling law)]
In the worst case $\left|W-y_u\right|\sim W$, the bound in \eqref{eq:Nmax} simplifies to
\[
M_{\max} = \mathcal O\!\left(\sqrt{\frac{R}{\lambda}}\right),
\]
which coincides with the classical \ac{nf} scaling.
This scaling holds because the reflector location is fully constrained by
specular geometry; treating it as a free parameter would generally introduce a first-order error term and lead to a significantly more restrictive bound.
\end{remark}

\begin{theorem}[\textbf{Generalized validity condition for a \ac{srnf} channel model}]
Consider a \ac{bs} equipped with a \ac{ura} with $M$ as the larger dimension of the array, with inter-element spacing $d=\lambda/2$, operating at wavelength $\lambda$.
Assume that the dominant \ac{nlos} path is generated by a single specular reflection on a planar surface $S$, such that the reflection point is uniquely determined by the geometry of the \ac{bs}, the \ac{ue}, and the reflecting plane.
%, and a \ac{ue} located at $(R,y_u,z_u)$.
% defined by $y=W$
%=(R,\,2W-y_u,\,z_u)

Let $\mathcal{H}$ denote the \ac{nf} channel observed across the \ac{bs} array, and let $\varepsilon^{}_\Phi$ be the maximum tolerable phase error.
Then, there exists a maximum number of \ac{bs} antennas $M_{\max}$ such that:
\begin{itemize}
    \item for $M \le M_{\max}$, the channel $\mathcal{H}$ can be equivalently represented by a single reflector located on the planar surface;
    \item for $M > M_{\max}$, he single-reflector representation is no longer guaranteed within the phase tolerance
    %the channel $\mathcal{H}$ can only be represented by distinct paths converging in a \ac{va} (e.g., a \ac{vue}), single-reflector approximation no longer holds.
\end{itemize}
Let $\mathbf p_v^{}$ denote the \ac{vue} location and $\rho \triangleq \|\mathbf p_v^{}\|$ the
corresponding propagation distance.
% \begin{equation}
% \label{eq:rho_3d}
% \rho \triangleq \|\mathbf p_v^{}\| = \sqrt{R^2+(2W-y_u)^2+z_u^2} \;.
% \end{equation}
% Let $\mathbf p_m$ denote the position of the $m$-th antenna, and let \begin{equation}
% \label{eq:aperp}
% r^\perp_{\max} \triangleq \max_{m\in\{0,\dots,M-1\}} \| \mathbf P^\perp_{\mathbf{U}} \mathbf p_m\|
% \end{equation}
% be the maximum array aperture measured along the direction orthogonal to the propagation direction $\mathbf u \triangleq \dfrac{\mathbf p_v^{}}{\rho}$, where $\mathbf U =\mathbf u \mathbf u^\mathsf{T}$.
%$\mathbf P_\perp \triangleq \mathbf I - \mathbf u \mathbf u^\mathsf{T}$.
Then, the threshold $M_{\max}$ is upper bounded by
%, given the projected edge aperture $r^{\perp}_{\max} \le (M-1)\,d$
% \begin{equation}
% \label{eq:Nmax_3d}
% r^\perp_{\max} \;\le\; \sqrt{\frac{\lambda\,\rho\,\varepsilon^{}_\Phi}}{\pi}}.
% \end{equation}
% \eqref{eq:Nmax_3d}, it implies
\begin{equation}
\label{eq:Nmax_3d_closed}
M_{\max}
\;\le\;
1 + 2\sqrt{\frac{\rho\,\varepsilon^{}_\Phi}{\lambda\,\pi}}.
\end{equation}
\end{theorem}

\begin{proof}
Using the \ac{vue} construction, the single-bounce \ac{nlos} path via the planar
reflector is equivalent to a \ac{los} path from the \ac{bs} to the \ac{vue} located at $\mathbf p_v^{}$.
% \[
% \mathbf p_v^{} = (R,\,2W-y_u,\,z_u).
% \]
Let $\mathbf p_0^{}=\mathbf{0}_{3\times1}$ 
denote the reference antenna position and $\mathbf p_m$ the $m$-th antenna position.
The exact \ac{nlos} path length to antenna $m$ is
\begin{equation}
d_{v,m} = \|\mathbf p_v^{} - \mathbf p_m^{}\|
%=\|\mathbf p_v^{}-\mathbf r_m\|.
\end{equation}
Denoting $\rho=\|\mathbf p_v^{}\|$, $\mathbf u=\dfrac{\mathbf p_v^{}}{\rho}$, and $\mathbf U =\mathbf u \mathbf u^\mathsf{T}$, a second-order Fresnel/Taylor expansion for $\|\mathbf p_m\|\ll \rho$ yields
\begin{equation}
\label{eq:taylor_3d}
d_{v,m}
\approx
\rho
-
\mathbf u^\mathsf{T} \mathbf p_m
+
\frac{1}{2\rho}
\left(
\|\mathbf p_m\|^2-(\mathbf u^\mathsf{T} \mathbf p_m)^2
\right) .
%+O\!\left(\frac{\|\mathbf p_m\|^3}{\rho^2}\right).
\end{equation}
Hence, the path difference satisfies
\begin{equation}
\label{eq:delta_star_3d}
\delta_m^\ast \triangleq d_{v,m}-d_{v,0}
\approx
-\mathbf u^\mathsf{T} \mathbf p_m
+
\frac{1}{2\rho}\|\mathbf P^\perp_\mathbf{U} \mathbf p_m\|^2.
\end{equation}

Now consider the \ac{srnf} model that enforces a \emph{single}
reflector point 
%on $\mathcal S$ (equivalently, a single specular point associated with the reference antenna) 
 to generate the reflector--\ac{bs} segment for all
array elements. Under planar specular geometry, the reflector location is uniquely constrained, and the first-order (linear) term in the modeling mismatch cancels across the array, leaving a residual phase-relevant mismatch that is quadratic in the projected aperture. Consequently, neglecting the higher-order terms, the dominant approximation error is given by
\begin{equation}
\label{eq:error_3d}
\left|\Delta \delta_m\right|
\;\approx\;
\frac{1}{2\rho}\|\mathbf P^\perp_\mathbf{U} \mathbf p_m\|^2.
\end{equation}
Let \begin{equation}
\label{eq:aperp}
r^\perp_{\max} \triangleq \max_{m\in\{0,\dots,M-1\}} \| \mathbf P^\perp_{\mathbf{U}} \mathbf p_m\|
\end{equation}
be the maximum array displacement projected onto the plane orthogonal to the propagation direction, then 
\begin{equation}
\left|\Delta \delta_m\right|
\;\lesssim\; \frac{(r^\perp_{\max})^2}{2\rho}
\;\leq\;
\frac{(M-1)^2\,d^2}{2\rho}.
\end{equation}
Imposing the phase error constraint, we obtain 
\[
\frac{2\pi}{\lambda}\,\frac{(M-1)^2\,d^2}{2\rho}\le \varepsilon^{}_\Phi.
\]
% and using \eqref{eq:error_3d} yields
% \[
% \|\mathbf P^\perp_\mathbf{U} \mathbf p_m\|^2 \le \frac{\lambda\,\rho\,\varepsilon^{}_\Phi}{\pi},
% \]
% which implies 
% \begin{equation}
% \label{eq:Nmax_3d}
% r^\perp_{\max} \;\le\; \sqrt{\frac{\lambda\,\rho\,\varepsilon^{}_\Phi}{\pi}}.
% \end{equation}
% For a linear array with spacing $d=\lambda/2$,
% the projected edge aperture satisfies $r^{\perp}_{\max} \le (M-1)\,d$, which gives
% \eqref{eq:Nmax_3d_closed} and completes the proof.
With $d=\lambda/2$, solving for $M$ gives \eqref{eq:Nmax_3d_closed}, which completes the proof.
\end{proof}

\subsection{Measurement and State Models}
We propose a vehicle localization procedure by tracking the dynamic state 
\begin{equation}
\mathbf s = \begin{bmatrix} \mathbf p_u^{\mathsf T} & y_{v,1}^{} & z_{v,1}^{} & \cdots & y_{v,\widehat{L}-1} & z_{v,\widehat{L}-1} \end{bmatrix}^{\mathsf T} \in \mathbb{R}^{D},
\end{equation} 
with $D=2\widehat{L}+1$, where $\widehat{L}$ denotes the total number of tracked propagation states, i.e., the \ac{ue} (index $0$) plus its tracked \acp{vue}. The state therefore includes the \ac{ue} 3D coordinates $\mathbf p_u^{}$ and, for each $\widehat{\ell}\in\{1,\dots,\widehat{L}-1\}$, the tuple $(y_{v,\widehat{\ell}}, z_{v,\widehat{\ell}})$ of the corresponding tracked \ac{vue}. Since the tunnel walls are perpendicular to the \ac{bs} orientation, the $x$-coordinates of the \ac{ue} and the \acp{vue} coincide. 
The measurement model is defined as 
\begin{equation}
    \mathbf h(\mathbf s) =
    \begin{bmatrix}
    {\bm \phi}^{\mathsf T} &
    {\bm \psi}^{\mathsf T} &
    {\bm \kappa}^{\mathsf T} &
    \Delta \bm d^{\mathsf T}
    \end{bmatrix}^{\mathsf T},
\end{equation}
with ${\bm \phi}=\begin{bmatrix}{\phi}_0^{} & \cdots & {\phi}_{\widehat{L}-1}\end{bmatrix}^\mathsf{T}$,
${\bm \psi}=\begin{bmatrix}{\psi}_0^{} & \cdots & {\psi}_{\widehat{L}-1}\end{bmatrix}^\mathsf{T}$, 
${\bm \kappa}=\begin{bmatrix}{\kappa}_0^{} & \cdots & {\kappa}_{\widehat{L}-1}\end{bmatrix}^\mathsf{T}$, and
$\Delta{\bm d}=\begin{bmatrix}\Delta{d}_1^{} & \cdots & \Delta d_{\widehat{L}-1}\end{bmatrix}^\mathsf{T}$, where $\Delta{d}_\ell^{} = d_\ell^{}-d_0^{}$ is the single anchor \ac{tdoa}, a clock offset unbiased measurement~\cite{SA-JURE}. 

Given the generic 3D \ac{vue} coordinates $\mathbf p^{}_v = \begin{bmatrix}x^{}_u & y^{}_v & z^{}_v\end{bmatrix}$, let $\Delta y = y^{}_v - y^{}_u$ and $\Delta z = z^{}_v - z^{}_u$. Defining the normal to the reflecting plane as $\vec{\bm n} = \frac{1}{\sqrt{\Delta y^2 + \Delta z^2}}\begin{bmatrix}0 & \Delta y & \Delta z\end{bmatrix}^\mathsf{T}$, the reflector associated with the \ac{vue} is given by
\begin{equation}
    \mathbf{p}^{}_r = \left(\frac{\vec{\bm n}^\mathsf{T}\mathbf p_u^{} + \vec{\bm n}^\mathsf{T}\mathbf p_v^{}}{2\vec{\bm n}^\mathsf{T}\mathbf p_v^{}}\right) \mathbf p_v^{}.
\end{equation}
The \ac{los} measurements are related to the \ac{ue} by
\begin{align}
    \phi_0^{} &= \tan^{-1} \left(\frac{y^{}_u}{x^{}_u}\right), 
    \\ 
    \psi_0^{} &= \tan^{-1} \left(\frac{z^{}_u}{\sqrt{x^{2}_u+y^{2}_u}}\right), 
    \\
    \kappa_0^{} &= \| \mathbf p^{}_{u}\|, 
\end{align}
while, the $\ell$-path measurements are related to the \ac{ue}, the \ac{vue}, and the reflector by
\begin{align}
    \phi_\ell &= \tan^{-1} \left(\frac{y^{}_v}{x^{}_u}\right), 
    \\ 
    \psi_\ell^{} &= \tan^{-1} \left(\frac{z^{}_v}{\sqrt{x^{2}_u+y^{2}_v}}\right), 
    \\
    \kappa_\ell^{} &= \|\mathbf p^{}_{r,\ell}\|, 
    \\     
    \Delta d_\ell^{} &= \| \mathbf p^{}_{v,\ell} \| - \|\mathbf p^{}_u\|.
\end{align}
% \begin{table*}[!t]%
% \centering %
% \caption{This is sample table caption.\label{tab2}}%
% \begin{tabular*}{\textwidth}{@{\extracolsep\fill}lllll@{\extracolsep\fill}}
% \toprule
% \textbf{Col1 head} & \textbf{Col2 head}  & \textbf{Col3 head}  & \textbf{Col4 head}  & \textbf{Col5 head} \\
% \midrule
% col1 text & col2 text  & col3 text  & col4 text  & col5 text\tnote{$^\dagger$}   \\
% col1 text & col2 text  & col3 text  & col4 text  & col5 text   \\
% col1 text & col2 text  & col3 text  & col4 text  & col5 text\tnote{$^\ddagger$}   \\
% \bottomrule
% \end{tabular*}
% \begin{tablenotes}
% \item[$^\dagger$] Example for a first table footnote.
% \item[$^\ddagger$] Example for a second table footnote.
% \item {\it Source}: Example for table source text.
% \end{tablenotes}
% \end{table*}
After the \ac{cp} decomposition, we obtain for each path $\ell$ the estimated parameter vector
\begin{equation}
    \widehat{\bm{\rho}}_\ell =
    \begin{bmatrix}
    \widehat{\alpha}_\ell &
    \widehat{\phi}_\ell &
    \widehat{\psi}_\ell &
    \widehat{\kappa}_\ell &
    \widehat{d}_\ell &
    \widehat{v}_\ell
    \end{bmatrix}^{\mathsf T}.
\end{equation}
To ensure measurement reliability, the following consistency checks are applied:
\emph{(i)} if $\widehat{d}_\ell < 0$, the measurement is discarded;  \emph{(ii)} if $\widehat{\kappa}_\ell < 0$, the \ac{ff} estimate is used; \emph{(iii)} otherwise, the \ac{nf} estimate is used.
After this validation step, the measurement vector for \ac{ue} positioning is constructed by stacking the estimated angular and range-related parameters of all valid paths as
\begin{equation}
    \mathbf z =
    \begin{bmatrix}
    \widehat{\bm \phi}^\mathsf{T} &
    \widehat{\bm \psi}^\mathsf{T} &
    \widehat{\bm \kappa}^\mathsf{T} &
    \Delta\widehat{\bm d}^\mathsf{T}
    \end{bmatrix}^{\mathsf T} \in \mathbb{R}^{K},
\end{equation}
with $\widehat{\bm \phi}=\begin{bmatrix}\widehat{\phi}_0^{} & \cdots & \widehat{\phi}_{L-1}^{}\end{bmatrix}^\mathsf{T}$,
$\widehat{\bm \psi}=\begin{bmatrix}\widehat{\psi}_0^{} & \cdots & \widehat{\psi}_{L-1}^{}\end{bmatrix}^\mathsf{T}$, 
$\widehat{\bm \kappa}=\begin{bmatrix}\widehat{\kappa}_0^{} & \cdots & \widehat{\kappa}_{L-1}^{}\end{bmatrix}^\mathsf{T}$, and
$\Delta\widehat{\bm d}=\begin{bmatrix}\Delta\widehat{d}_1^{} & \cdots & \Delta\widehat{d}_{L-1}^{}\end{bmatrix}^\mathsf{T}$, in which $\Delta\widehat{d}_\ell^{} = \widehat{d}_\ell^{}-\widehat{d}_0^{}$.
Algorithm~\ref{alg:meas_extr} summarizes the measurement extraction and sanitization procedure.
\begin{algorithm}[t]
\caption{\textsc{MeasurementExtraction}}
\begin{algorithmic}[1]
% \STATE $\{\widehat{\bm{\rho}}_\ell\} \leftarrow \mathrm{CPD}(\mathcal{Y}_k, L)$
% \STATE $\mathcal{L}_k = \{\ell : \widehat{d}_\ell \ge 0\}$
% \STATE Select NF/FF estimates
% \STATE $\mathbf{z}_k = [\widehat{\boldsymbol\phi},\widehat{\boldsymbol\psi},\widehat{\boldsymbol\kappa},\Delta\widehat{\mathbf d}]$
\Require $\tilde{\mathcal{H}},\,L$

\State $\{\widehat{\mathbf{B}}_s,\widehat{\mathbf{B}}_f,\widehat{\mathbf{B}}_t\}\leftarrow \mathrm{CPD}(\tilde{\mathcal{H}},L)$
\State estimate $\{\widehat{\alpha}_\ell,\widehat{d}_\ell,\widehat{v}_\ell,\widehat{\phi}^{\mathrm{NF}}_\ell,\widehat{\psi}^{\mathrm{NF}}_\ell,\widehat{\kappa}^{\mathrm{NF}}_\ell,\widehat{\phi}^{\mathrm{FF}}_\ell,\widehat{\psi}^{\mathrm{FF}}_\ell\}_{\ell=0}^{L-1}$

    \For{$\ell=0,\ldots,L-1$}
        \If{$\widehat{d}_\ell < 0$}
            \State discard path $\ell$
        \ElsIf{$\widehat{\kappa}^{\mathrm{NF}}_\ell < 0$}
            \State $\widehat{\phi}_\ell \leftarrow \widehat{\phi}^{\mathrm{FF}}_\ell,\;
                   \widehat{\psi}_\ell \leftarrow \widehat{\psi}^{\mathrm{FF}}_\ell$
        \Else
            \State $\widehat{\phi}_\ell \leftarrow \widehat{\phi}^{\mathrm{NF}}_\ell,\;
                   \widehat{\psi}_\ell \leftarrow \widehat{\psi}^{\mathrm{NF}}_\ell,\;
                   \widehat{\kappa}_\ell \leftarrow \widehat{\kappa}^{\mathrm{NF}}_\ell$
        \EndIf
    \EndFor

    \State $\Delta\widehat d_\ell \leftarrow \widehat d_\ell-\widehat d_0,\;\forall \ell\geq 1$
    \State $\mathbf{z} \leftarrow [\widehat{\boldsymbol\phi}^{\mathsf T} \, \widehat{\boldsymbol\psi}^{\mathsf T} \,\widehat{\boldsymbol\kappa}^{\mathsf T}\,\Delta\widehat{\mathbf d}^{\mathsf T}]^{\mathsf T}$
    \State \Return $\mathbf{z}$
\end{algorithmic}
\label{alg:meas_extr}
\end{algorithm}

%\subsection{Data Association}
\subsection{Measurement -- State Track Association}
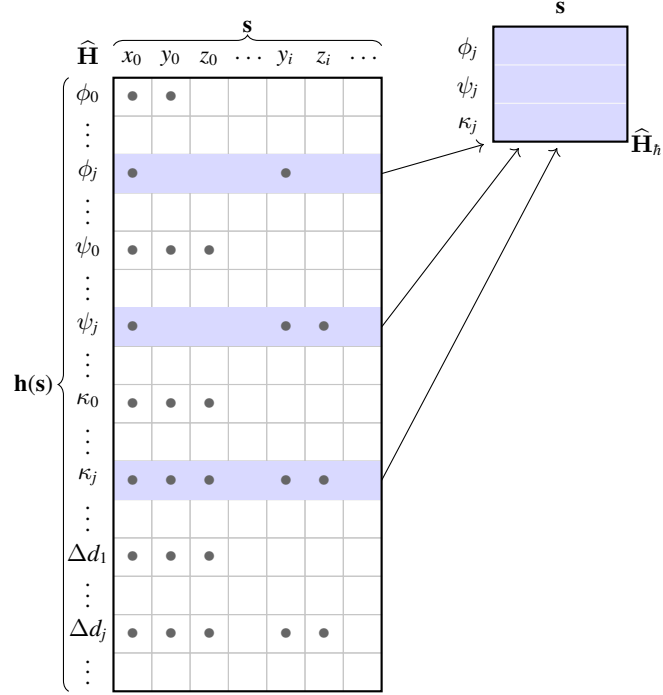
\begin{figure}[t]
    \centering
    \begin{tikzpicture}[x=0.42cm,y=0.42cm,font=\scriptsize]
        % Matrix size
        \def\W{8.4}
        \def\H{19.2}

        % Main Jacobian grid (one square per index pair)
        %\draw[thick] (0,0) rectangle (\W,\H);
        %\draw[step=0.4cm,gray!35] (0,0) grid (\W,\H);

        % Grid square size and automatic placement of row indices
        \def\wy{1.2}
        \def\wx{1.2}
        \pgfmathsetmacro{\Wx}{\W-\wx}
        \pgfmathsetmacro{\Hy}{\H-\wy}
        \foreach \y in {0,\wy,...,\Hy}
        \foreach \x in {0,\wx,...,\Wx}{\draw[gray!45] (\x,\y) rectangle (\x+\wx,\y+\wy);}

        % Measurement/state-row index labels (top-to-bottom), auto-spaced with \wy
        \foreach [count=\i from 0] \lbl in {$\phi_0$,$\phi_j$,$\psi_0$,$\psi_j$,$\kappa_0$,$\kappa_j$,$\Delta d_1$,$\Delta d_j$} {
            \pgfmathsetmacro{\yy}{\H - 0.5 - \i*2*\wy}
            \node[align=right] at (-0.8,\yy) {\lbl};
        }
        \foreach \i in {0,...,7}{
            \pgfmathsetmacro{\yy}{\H - 0.35 - \wy - \i*2*\wy}
            \node[align=right] at (-0.8,\yy) {$\vdots$};
        }

        \draw[decorate,decoration={brace,amplitude=4pt}] (-1.4,0) -- (-1.4,\H);
        \node at (-2.5,8*\wy) {$\textbf{h}(\textbf{s})$};
        % {$\phi_0$,$\vdots$,$\phi_j$,$\vdots$,$\psi_0$,$\vdots$,$\psi_j$,$\vdots$,$\kappa_0$,$\vdots$,$\kappa_j$,$\vdots$,$\Delta d_1$,$\vdots$,$\Delta d_j$,$\vdots$}

        % Header squares for state indices (top)
        %\foreach \x in {0,1.2,...,16} %{\draw[gray!45] (\x,\H) rectangle (\x+1.2,\H+1.2);}
        \node at (\wx/2,\H+0.6) {$x_0$};
        \node at (\wx/2+\wx,\H+0.6) {$y_0$};
        \node at (\wx/2+2*\wx,\H+0.6) {$z_0$};
        \node at (\wx/2+0.1+3*\wx,\H+0.6) {$\cdots$};
        \node at (\wx/2+4*\wx,\H+0.6) {$y_i$};
        \node at (\wx/2+5*\wx,\H+0.6) {$z_i$};
        \node at (\wx/2+0.1+6*\wx,\H+0.6) {$\cdots$};
        % Curly bracket grouping all state variables on top
        \draw[decorate,decoration={brace,amplitude=4pt}] (0,\H+1.0) -- (7*\wx,\H+1.0);
        \node at (3.5*\wx,\H+1.6) {$\textbf{s}$};
        % \node at (\wx/2+0.2+7*\wx,\H+0.5) {$y_{v_{\widehat{L}-1}}$};
        % \node at (\wx/2+0.2+8*\wx,\H+0.5) {$z_{v_{\widehat{L}-1}}$};

        % Matrix label
        \node at (-0.8,\H+0.8) {$\widehat{\mathbf{H}}$};

        % Highlight exactly the rows of ell measurements: phi_ell, psi_ell, kappa_ell
        \fill[blue!15] (0,13*\wy) rectangle (\W,14*\wy);   % phi_ell row
        \fill[blue!15] (0,9*\wy) rectangle (\W,10*\wy);    % kappa_ell row
        \fill[blue!15] (0,5*\wy) rectangle (\W,6*\wy);    % psi_ell row

        % Re-draw grid lines on top of highlights
        %\draw[step=1cm,gray!35] (0,0) grid (\W,\H);
        \draw[thick] (0,0) rectangle (\W,\H);

        \node (Hl) at (\W+3.5,\H-2){};
        \draw[thick, fill=blue!15] (Hl) rectangle ($(Hl)+(3.5*\wx, 3*\wy)$);
        \node at ($(Hl)+(4*\wx, 0)$) {$\widehat{\mathbf{H}}_{\hbar}$};
        \draw[->] (\W,14*\wy-\wy/2)--(Hl);
        \draw[->] (\W,10*\wy-\wy/2)--($(Hl)+(0.8,-0.2)$);
        \draw[->] (\W,6*\wy-\wy/2)--($(Hl)+(2,-0.2)$);

        \node at ($(Hl)+(3.5*\wx/2,3.5*\wy)$) {$\textbf{s}$};
        \node at ($(Hl)+(-0.8,0.5)$) {$\kappa_j$};
        \node at ($(Hl)+(-0.8,0.5+\wy)$) {$\psi_j$};
        \node at ($(Hl)+(-0.8,0.5+2*\wy)$) {$\phi_j$};
        \draw [gray!5] ($(Hl)+(0.05, \wy)$)--($(Hl)+(3.5*\wx-0.05, \wy)$);
        \draw [gray!5] ($(Hl)+(0.05, 2*\wy)$)--($(Hl)+(3.5*\wx-0.05, 2*\wy)$);

        \foreach \i in {1,3,...,15}{
            \pgfmathsetmacro{\yy}{\wy/2+\i*\wy}
            \node[black!60] at (\wx/2,\yy) {$\bullet$};
        }
        %Phi
        \node[black!60] at (\wx/2+\wx,15*\wy+\wy/2) {$\bullet$};
        \node[black!60] at (\wx/2+4*\wx,13*\wy+\wy/2) {$\bullet$};
        %Psi
        \node[black!60] at (\wx/2+\wx,11*\wy+\wy/2) {$\bullet$};    
        \node[black!60] at (\wx/2+2*\wx,11*\wy+\wy/2) {$\bullet$};
        \node[black!60] at (\wx/2+4*\wx,9*\wy+\wy/2) {$\bullet$};    
        \node[black!60] at (\wx/2+5*\wx,9*\wy+\wy/2) {$\bullet$}; 
        %Kappa
        \node[black!60] at (\wx/2+\wx,7*\wy+\wy/2) {$\bullet$};    
        \node[black!60] at (\wx/2+2*\wx,7*\wy+\wy/2) {$\bullet$};    
        \node[black!60] at (\wx/2+\wx,5*\wy+\wy/2) {$\bullet$};    
        \node[black!60] at (\wx/2+2*\wx,5*\wy+\wy/2) {$\bullet$}; 
        \node[black!60] at (\wx/2+4*\wx,5*\wy+\wy/2) {$\bullet$};    
        \node[black!60] at (\wx/2+5*\wx,5*\wy+\wy/2) {$\bullet$}; 
        %Dd
        \node[black!60] at (\wx/2+\wx,3*\wy+\wy/2) {$\bullet$};    
        \node[black!60] at (\wx/2+2*\wx,3*\wy+\wy/2) {$\bullet$};    
        \node[black!60] at (\wx/2+\wx,\wy+\wy/2) {$\bullet$};    
        \node[black!60] at (\wx/2+2*\wx,\wy+\wy/2) {$\bullet$}; 
        \node[black!60] at (\wx/2+4*\wx,\wy+\wy/2) {$\bullet$};    
        \node[black!60] at (\wx/2+5*\wx,\wy+\wy/2) {$\bullet$};

        % Row tags
        % \node[right] at (\W+0.2,13*\w+0.5) {$\phi_i$ row};
        % \node[right] at (\W+0.2,9*\w+0.5) {$\psi_i$ row};
        % \node[right] at (\W+0.2,5*\w+0.5) {$\kappa_i$ row};
    \end{tikzpicture}
    \caption{Visualization of the Jacobian matrix structure used for data association, highlighting the $(i,j)$ track--measurement pair and the corresponding submatrix construction. The black dots represent the non-zero measurement contributions to the respective state elements.}
    \label{fig:H_matrix_grid}
\end{figure}
To associate each measurement with the corresponding state track, we adopt a gated \ac{nn} strategy based on the Mahalanobis distance.
Let $\widehat{\mathbf{s}}$ and $\widehat{\mathbf{P}}$ denote the predicted state and covariance at the current time step; for simplicity, here the time index is omitted. The measurement model is linearized around $\widehat{\mathbf{s}}$ through the Jacobian
\begin{equation}
    \widehat{\mathbf{H}} = \left.\frac{\partial \mathbf{h}(\mathbf{s})}{\partial \mathbf{s}}\right|_{\substack{\boldsymbol{s}=\widehat{\boldsymbol{s}}}} \in \mathbb{R}^{K \times D}.
\end{equation}
Figure~\ref{fig:H_matrix_grid} illustrates the structure of $\widehat{\mathbf{H}}$ and the row-selection mapping used to extract $\widehat{\mathbf{H}}_{\hbar}$ for each candidate pair.

For association, we evaluate every candidate track--measurement pair. Let $i\in \{0,\dots,\widehat{L}-1\}$ denote a predicted track and $j$ a candidate measurement vector. For measurement $j\in\{0,\dots,L-1\}$, we extract the three rows associated with $\big(\phi_j,\,\psi_j,\,\kappa_j\big)$ and express their row indices with $\hbar$, yielding the submatrix
\(
    \widehat{\mathbf{H}}_{\hbar} \in \mathbb{R}^{3 \times D}.
\)
The innovation covariance for pair $(i,j)$ is
\begin{equation}
    \mathbf{S}_{i,j} = \widehat{\mathbf{H}}_{\hbar}^{}\,\widehat{\mathbf{P}}_i^{}\,\widehat{\mathbf{H}}_{\hbar}^{\mathsf T} + \mathbf{R}_{\hbar}^{} \in \mathbb{R}^{3\times 3},
\end{equation}
where $\mathbf{R}_{\hbar} \in \mathbb{R}^{3\times 3}$ is the measurement noise covariance matrix associated with $\big(\phi_j,\,\psi_j,\,\kappa_j\big)$. The corresponding innovation vector is
\begin{equation}
    \mathbf{y}_{i,j}^{} = \mathbf{z}_{\hbar}^{} - \mathbf{h}_{\hbar}^{(i)}(\widehat{\mathbf{s}}),
\end{equation}
where $\mathbf{h}_{\hbar}^{(i)}$ selects the $\hbar$ measurements of the $i$-track of the measurement model,
and the squared Mahalanobis distance is
\begin{equation}
    \eta_{i,j}^{2} = \mathbf{y}_{i,j}^{\mathsf T}\,\mathbf{S}_{i,j}^{-1}\,\mathbf{y}_{i,j}.
\end{equation}

Under Gaussian assumptions, $\eta_{i,j}^2$ follows a chi-square distribution with three degrees of freedom. A validation gate is therefore defined as
\begin{equation}
    \eta_{i,j}^2 < \chi^{2}_3(p),
\end{equation}
where $p$ is the gating probability. Only track--measurement pairs satisfying this condition are considered feasible. Among the validated candidates, the final association is obtained by selecting the pair with minimum Mahalanobis distance, resulting in a maximum-likelihood consistent \ac{nn} assignment.
Algorithm~\ref{alg:data_assoc} reports the data association procedure.
\begin{algorithm}[t]
\caption{\textsc{MeasurementAssociation}}
\begin{algorithmic}[1]
\Require $\mathbf{z},\,\widehat{\mathbf{s}},\,\widehat{\mathbf{P}},\,p$
\State $\widehat{\mathbf{H}} = \left.\frac{\partial \mathbf{h}}{\partial \mathbf{s}}\right|_{\substack{\mathbf s =\widehat{\mathbf s}}}$
\For{each $(i,j)$}
    \State $\hbar \leftarrow \textsc{Index}(\mathbf h, (\phi_j,\psi_j,\kappa_j))$
    \State $\mathbf{y}_{i,j}^{} \leftarrow \mathbf{z}_{\hbar}^{} - \mathbf{h}_{\hbar}^{(i)} (\widehat{\mathbf{s}})$
    \State $\mathbf{S}_{i,j} \leftarrow \widehat{\mathbf{H}}_{\hbar}^{} \widehat{\mathbf{P}}\widehat{\mathbf{H}}_{\hbar}^\mathsf{T} + \mathbf{R}_{\hbar}^{}$
    \State $\eta^2_{i,j} \leftarrow \mathbf{y}_{i,j}^\mathsf{T} \mathbf{S}_{i,j}^{-1} \mathbf{y}_{i,j}^{}$
\EndFor
\State $\mathcal{O} \leftarrow \{(i,j): \eta^2_{i,j} < \chi^2_3(p)\}$
\State \Return $\mathcal{O}$
\end{algorithmic}
\label{alg:data_assoc}
\end{algorithm}

\subsection{Track Management}
The pairing procedure returns the sets of associated and unassociated indices for both measurements and state tracks. 
Tracks associated with measurements are updated, while unassociated tracks are maintained but not updated. Specifically, unassociated tracks are kept alive up to a predefined maximum number of consecutive missed associations $\zeta$. This is equivalent to deterministically assigning a track existence probability, which is equal to $1$ for associated and updated tracks, and linearly decreases to $0$ after $\zeta$ consecutive time steps without association. When this condition is met, the track is removed from the state vector, with the exception of the \ac{ue} state, which is always preserved.
The unassociated measurement indices are instead used to initialize new tracks. 
The corresponding state estimate is initialized by exploiting the associated $\hbar$ measurement parameters and the current \ac{ue} estimate $\widehat{x}_u$, with an appropriately large initial covariance to account for initialization uncertainty.

After the birth-death management step, the associated measurement indices are reordered to ensure consistency with the state vector structure. In particular, the measurement associated with the first track (corresponding to the \ac{ue}) is labeled as $\ell_{\text{LoS}}$.
If the \ac{ue} state is not associated with any measurement, the LoS path is identified according to the geometric consistency condition
\begin{equation}
    \ell_{\text{LoS}} = \arg\min_\ell \left|\widehat{d}_\ell - \widehat{\kappa}_\ell \right|,
\end{equation}
subject to
\begin{equation}
     1-\gamma \leq 
     \frac{\widehat{d}_{\ell_{\text{LoS}}}}{\widehat{\kappa}_{\ell_{\text{LoS}}}}
     \leq 1+\gamma,
     \label{eq:los_condition}
\end{equation}
where $0 < \gamma \leq 1$ is a design threshold.
If condition~\eqref{eq:los_condition} is not satisfied, the scenario is treated as \ac{nlos}.
Moreover, by exploiting the Doppler shift $v_\ell$ associated with each propagation path, it is possible to discriminate between static and dynamic reflectors. Given the ego-vehicle velocity, the Doppler contribution of paths reflected by static objects can be predicted. Therefore, paths whose Doppler is consistent with this prediction are associated with static reflectors, whereas significant deviations indicate dynamic reflectors, which are treated as clutter.  While these paths can still contribute to localization at the current time step, they are not propagated to subsequent ones.
Algorithm~\ref{alg:track_manage} outlines the comprehensive track management approach.
\begin{algorithm}[t]
\caption{\textsc{TrackManagement}}
\begin{algorithmic}[1]
    \Require $\mathbf z,\, \widehat{\mathbf s},\, \widehat{\mathbf P},\, \mathcal{O},\, \mathbf P_{0\left|0\right.},\,\zeta,\,\gamma$
    \For{each track $i \notin \mathcal{O}$}
        \State $c_i \leftarrow c_i+1$
        \If{$c_i\geq \zeta$ and $i\neq 0$}
            \State remove track $i$ from $\widehat{\mathbf{s}}$ and $\widehat{\mathbf{P}}$
        \EndIf
    \EndFor

    \For{each measurement $j\notin \mathcal{O}$}
        \State $\widehat{\mathbf{s}}^{\mathrm{new}} \leftarrow \textsc{Intersection}(\widehat{\phi}_j,\widehat{\psi}_j, \widehat{x}_{u})$
        \State append $\widehat{\mathbf{s}}^{\mathrm{new}}$ to $\widehat{\mathbf{s}}$
        \State append $\mathbf{P}_{0\left|0\right.}$ to $\widehat{\mathbf{P}}$
    \EndFor

    \If{$(0,j) \in \mathcal{O}$}
        \State $\ell_{\mathrm{LoS}}\leftarrow j$
        \State $\mathrm{LoS}\leftarrow 1$ 
    \Else
        \State $\ell_{\mathrm{LoS}}\leftarrow \arg\min_{\ell}\left|\widehat d_\ell-\widehat\kappa_\ell\right|$
        \If{$1-\gamma \leq \dfrac{\widehat d_{\ell_{\mathrm{LoS}}}}{\widehat\kappa_{\ell_{\mathrm{LoS}}}} \leq 1+\gamma$}
            \State $\mathrm{LoS}\leftarrow 1$ 
        \Else
            \State $\mathrm{LoS}\leftarrow 0$ 
        \EndIf
    \EndIf
    \State $\mathbf z \leftarrow \textsc{Sort}(\mathbf z,\mathcal{O})$ 
    \State \Return $\mathbf{z},\, \widehat{\mathbf s},\, \widehat{\mathbf P},\, \mathrm{LoS}$
\end{algorithmic}
\label{alg:track_manage}
\end{algorithm}

\subsection{Adaptive Tracking Filter}
The tracking filter is implemented as an \ac{ekf} with a variable state dimension, adapting to the dynamic birth and death of tracks. The state vector $\mathbf{s}_k$ at time step $k$ includes the \ac{ue} position and the set of active \acp{vue}. 
The filter operates in two stages, namely prediction and update.
\subsubsection*{Prediction}
The state evolution is modeled as
\begin{equation}
\widehat{\mathbf{s}}_{k\left|k-1\right.} = \mathbf{f}(\widehat{\mathbf{s}}_{k-1\left|k-1\right.}, \nu_{k-1},\theta_{k-1}) + \mathbf{w}_k,
\end{equation}
where $f(\cdot)$ is a non-linear function describing the state evolution, $\nu_{k-1}$ and $\theta_{k-1}$ are the speed and heading, respectively, and $\mathbf{w}_k \sim \mathcal{N}(\mathbf{0},\mathbf{Q}_k)$ is the process noise.
For the \ac{ue}, a velocity sensor model is adopted, while each \ac{vue} follows a random walk model~\cite{Gustafsson2005}. Accordingly, the predicted covariance is given by
\begin{equation}
\widehat{\mathbf{P}}_{k\left|k-1\right.}^{} = \mathbf{F}_k^{} \widehat{\mathbf{P}}_{k-1\left|k-1\right.}^{} \mathbf{F}_k^{\mathsf T} + \mathbf{Q}_k^{},
\end{equation}
where $\mathbf{F}_k^{}$ is the state transition Jacobian.

\subsubsection*{Update}

Given the measurement vector $\mathbf z_k^{}$ and the nonlinear measurement model described in the previous section, the innovation is computed as
\begin{equation}
\mathbf{y}_k^{} = \mathbf z_k^{} - \mathbf{h}(\widehat{\mathbf{s}}_{k\left|k-1\right.}^{}).
\end{equation}
The innovation covariance is
\begin{equation}
\mathbf{S}_k = \widehat{\mathbf{H}}_k^{} \widehat{\mathbf{P}}_{k\left|k-1\right.}^{} \widehat{\mathbf{H}}_k^{\mathsf T} + \mathbf{R}_k^{},
\end{equation}
where $\mathbf{R}_k^{} \in \mathbb{R}^{K\times K}$ is the measurement noise covariance matrix.
The Kalman gain is then given by
\begin{equation}
\mathbf{K}_k^{} = \widehat{\mathbf{P}}_{k\left|k-1\right.}^{} \widehat{\mathbf{H}}_k^{\mathsf T} \mathbf{S}_k^{-1},
\end{equation}
and the state and covariance are updated as
\begin{align}
\widehat{\mathbf{s}}_{k\left|k\right.}^{} &= \widehat{\mathbf{s}}_{k\left|k-1\right.}^{} + \mathbf{K}_k^{} \mathbf{y}_k^{}, \\
\widehat{\mathbf{P}}_{k\left|k\right.}^{} &= (\mathbf{I} - \mathbf{K}_k^{} \widehat{\mathbf{H}}_k^{})\widehat{\mathbf{P}}_{k\left|k-1\right.}(\mathbf{I} - \mathbf{K}_k^{} \widehat{\mathbf{H}}_k^{})^{\mathsf T} + \mathbf{K}_k^{} \mathbf{R}_k^{} \mathbf{K}_k^{\mathsf T}.
\end{align}
Due to the birth and death processes, the state dimension varies over time. Track removal is performed by marginalizing the corresponding components from $\widehat{\mathbf{s}}_k$ and $\widehat{\mathbf{P}}_k$, while newly initialized tracks are appended with appropriate covariance initialization. This results in a flexible filtering structure capable of adapting to the time-varying multipath environment.

Algorithm~\ref{alg:javelin} summarizes the complete JAVELIN pipeline comprising measurement extraction,  data association, track management, and the variable-dimension \ac{ekf} recursion.
\begin{algorithm}[t]
\caption{JAVELIN}
\begin{algorithmic}[1]
\Require $\mathcal{X},\,\widehat{\mathbf{s}}_{0\left|0\right.},\,\widehat{\mathbf{P}}_{0\left|0\right.},\,p,\,\zeta,\,\gamma$
\For{each time step $k$}
    \State obtain $\mathcal{Y}_k$, $L_k$, $\nu_k$, $\theta_k$
    \State compute $\mathbf{Q}_k$, $\mathbf{R}_k$
    \State $\tilde{\mathcal{H}}_k \leftarrow \textsc{ChannelEstimation}(\mathcal{Y}_k, \mathcal{X})$
    \State $\mathbf{z}_k \leftarrow \textsc{MeasurementExtraction}(\tilde{\mathcal{H}}_k, L_k)$
    \State $(\widehat{\mathbf{s}}_{k\left|k-1\right.},\widehat{\mathbf{P}}_{k\left|k-1\right.})\leftarrow$ $\textsc{Predict}(\widehat{\mathbf{s}}_{k-1\left|k-1\right.}, \nu_{k-1}, \theta_{k-1},\widehat{\mathbf{P}}_{k-1\left|k-1\right.}, \mathbf{Q}_k)$
    \State $\mathcal{O}_k \leftarrow \textsc{MeasurementAssociation}(\mathbf{z}_k,\widehat{\mathbf{s}}_{k\left|k-1\right.},\widehat{\mathbf{P}}_{k\left|k-1\right.},\mathbf{R}_k,p)$
    \State $(\mathbf{z}_k, \widehat{\mathbf{s}}_{k\left|k-1\right.},\widehat{\mathbf{P}}_{k\left|k-1\right.},  \mathrm{LoS})\leftarrow$ $\textsc{TrackManagement}(\mathbf{z}_k,\widehat{\mathbf{s}}_{k\left|k-1\right.},\widehat{\mathbf{P}}_{k\left|k-1\right.}, \mathcal{O}_k, \widehat{\mathbf{P}}_{0\left|0\right.}, \zeta, \gamma)$
    %\State $\ell_{\mathrm{LoS}} \leftarrow \textsc{LoSDetection}(\mathbf{z}_k,\gamma)$
    \State $(\widehat{\mathbf{s}}_{k\left|k\right.},\widehat{\mathbf{P}}_{k\left|k\right.}) \leftarrow \textsc{Update}(\mathbf{z}_k^{}, \widehat{\mathbf{s}}_{k\left|k-1\right.},\widehat{\mathbf{P}}_{k\left|k-1\right.}, \mathbf{R}_k, \mathrm{LoS})$
\EndFor
\end{algorithmic}
\label{alg:javelin}
\end{algorithm}

\begin{figure}[t]
\centering
    \subfloat[\label{fig:s1}]
{
\begin{tikzpicture}
\filldraw[draw=black,fill=black!10]  (-3.5,0) rectangle (4,1);
\draw[dashed, draw=white, line width=1pt](-3.5,0.5) -- (4,0.5);
\node[thick,draw=amaranth,fill=amaranth!20,minimum width=0.025cm, minimum height=0.45cm, inner sep=2pt] (BS) at(0.1,0.5){};
\draw[-{latex}]  (BS.east) -- (0.7,0.5);
\draw[-{latex}]  (BS.west) -- (-0.5,0.5);
\draw[{latex}-{latex}]  (-3.5,1.15) -- node[above,midway]{ 50 m} (0.1,1.15);
\draw[{latex}-{latex}]  (0.1,1.15) -- node[above,midway]{ 50 m} (4,1.15);
\node[]at(0.1,0.855){\footnotesize BS};
\node[thick,draw=black,fill=gray!40,circle,scale=0.35](z)at(-3.1,0.6){};
\node[] at ($ (z) - (0.21,-0.25) $){$z$};
\draw[-{Triangle}]  (z) -- node[above, midway, black]{$x$} ++(0:0.5);
\draw[-{Triangle}]  (z) -- node[left, midway, black]{$y$} ++(-90:0.5);
\end{tikzpicture}
}\\
\subfloat[\label{fig:s2}]{
\begin{tikzpicture}
    \node[]at(0,0){\includegraphics[width=0.6\linewidth]{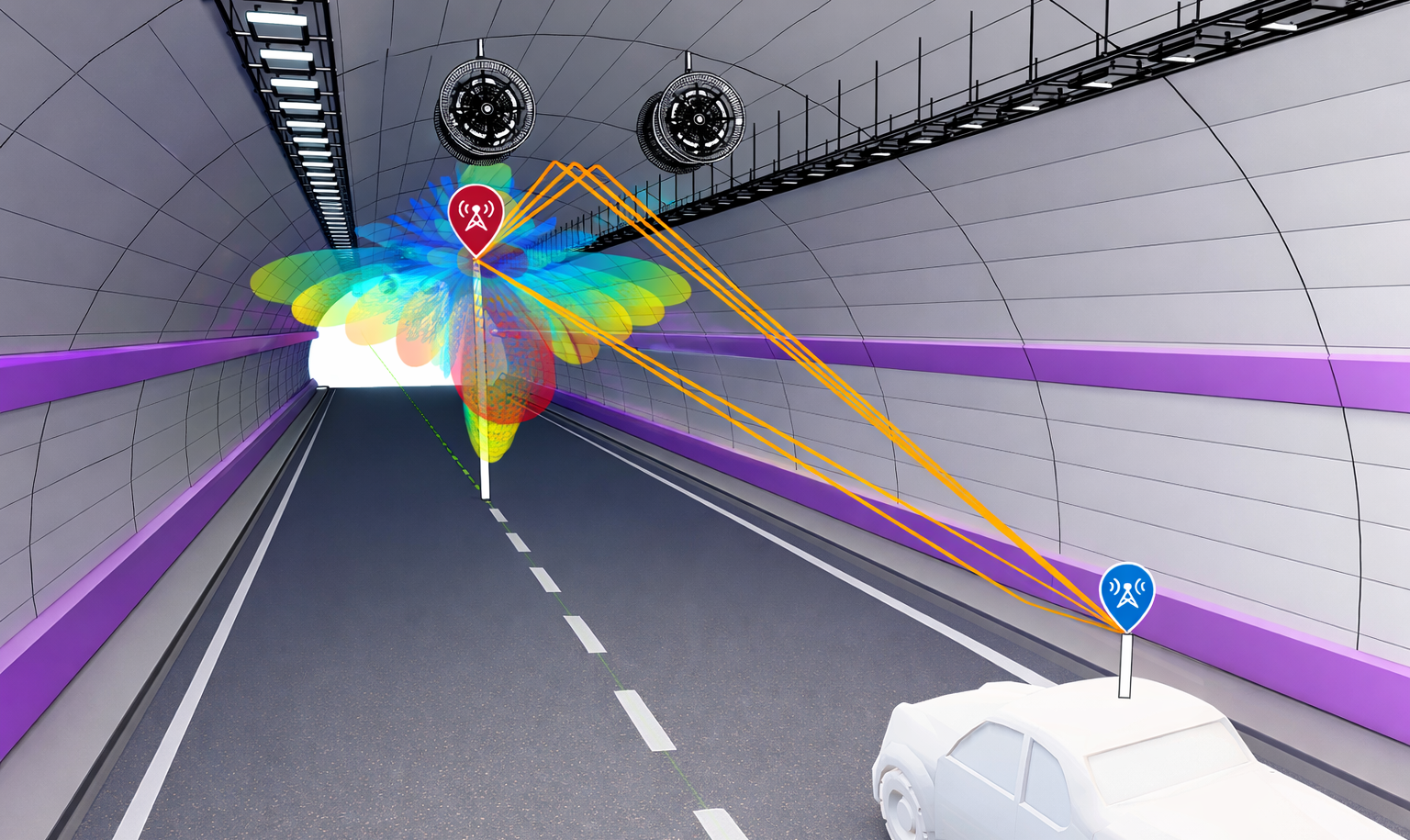}};
    \node[]at(2.6,-1.6){UE};
    \node[]at(-2.1,0.9){BS};
\end{tikzpicture}}
\caption{Implemented tunnel scenario. (a) Top view. (b) 3D representation in \textit{MATLAB Site Viewer} with a raytracer example. The red marker denotes a BS, the blue marker denotes the vehicle, and the purple stripes denote the RRMs. Beamforming gains for broadside and multipath propagation are also shown.}
\label{fig:simulation_scenario}
\end{figure}
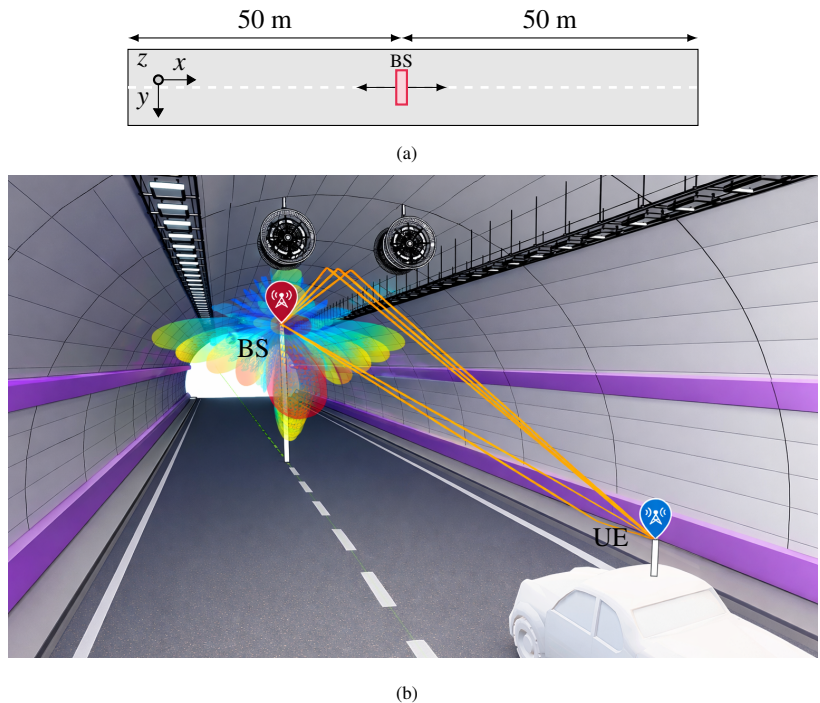
\section{Performance Evaluation}
\label{sec:results}
\subsection{Simulation Scenario}
To validate the proposed framework, we consider a realistic vehicular tunnel modeled in \textit{Blender}$\textsuperscript{\textregistered}$, illustrated in Figure~\ref{fig:simulation_scenario}, featuring a straight semi-cylindrical geometry with a length of $100\,\mathrm{m}$, a width of $10\,\mathrm{m}$, and a height of $5\,\mathrm{m}$. The \ac{bs} is placed at the center of the tunnel at a height of $4.8$\,m, equipped with two antenna arrays oriented in opposite directions, namely $(-90,-30)$\,deg and $(90,-30)$\,deg. 
The tunnel environment is then imported into \textit{MATLAB}$\textsuperscript{\textregistered}$~\cite{matlab} via the \textit{Site Viewer}, where the \ac{ul} \ac{srs} is simulated using the \textit{5G Toolbox} and a raytracer with a single bounce. Additionally, four metallic \acp{rrm} are deployed within the tunnel: two are positioned at the junction between the sidewalk and the wall with an inclination of $55$\,deg, and two are mounted on the side walls at a height of $3.3$\,m, with a $90$\,deg orientation, located midway between the vehicle and the anchor. 
%The tunnel scenario is illustrated in Figure~\ref{fig:simulation_scenario}. 
Velocities and the trajectories of the ego vehicle are generated using the \textit{MATLAB Driving Scenario}.

\subsection{Simulation Parameters}
The simulated cellular \ac{v2x} communication system operates at a carrier frequency of $f_c = 5.9$\,GHz, with a signal bandwidth of $100$\,MHz and a transmit power of $23$\,dBm. This bandwidth choice reflects a forward-looking scenario in which wider channel allocations at $5.9$\,GHz may become available, in line with the evolution towards \ac{6g}, which is expected to support higher data-rate demanding services.
The channel is modeled as in \eqref{eq:channel_model}, accounting for time-varying multipath fading. The noise power is modeled as $N_0 = k_B \cdot BW \cdot T_e$, where $k_B$ is the Boltzmann constant, $BW$ is the bandwidth, and $T_e = T_{\text{ant}} + 290\,(N_{\text{F}} - 1)$ is the equivalent noise temperature. Here, $T_{\text{ant}} = 298$\,K denotes the antenna temperature, and $N_{\text{F}} = 5$\,dB is the noise figure~\cite{tr138857}.
The number of antennas is defined as $M = M_{\max}^2$, where $M_{\max}$ is selected according to the upper bound in \eqref{eq:Nmax}. 
For $R = 3.5$\,m, $W = 3$\,m, $y_u = 2.5$\,m, and $\varepsilon_{\Phi} = 0.15$\,rad, we obtain $M = M_{\max}^2 = 9.88^2 \approx 100$. The \ac{srs} is configured according to 3GPP Rel-16 positioning specifications, using 12 symbols per slot and a comb size of 8~\cite{italiano2025tutorial}. 

For the JAVELIN algorithm, the parameters are set to $p = 0.99$, $\zeta = 1$, and $\gamma = 0.2$.
The uncertainty parameters are defined as $\sigma_\phi=2$\,deg, $\sigma_\psi=2$\,deg, $\sigma_\kappa=1.5$\,m, $\sigma_{\Delta d}=1.5$\,m, $\sigma_\nu = 0.2$\,m/s, $\sigma_\theta = 1$\,deg, and $\sigma_P^{}=10$\,m, with the initial covariance matrix given by $P^{}_{0\left|0\right.}=\sigma_P^2\mathbf{I}^{}_D$.
The clock bias is modeled as a truncated Gaussian distribution with variance $50$\,ns and support in the interval $[-100, 100]$\,ns~\cite{camajori2023feasibility}.
We assume that the initial position (in open-sky conditions), the vehicle velocity along the tunnel, and the vehicle height are available at the \ac{bs}, which is a realistic assumption in a \ac{c-its} context leveraging \ac{cam}~\cite{cam}. In addition, the number of paths $L$ is assumed to be known, in accordance with the 3GPP Rel-17 standard~\cite{italiano2025tutorial}.

\subsection{Simulation Analyses and Results}
We evaluate the proposed framework under different visibility conditions and vehicular trajectories (straight motion and slalom maneuvers), and compare its performance with the TeNFiLoc algorithm~\cite{TeNFiLoc}. TeNFiLoc is used as a baseline since it represents a state-of-the-art snapshot NF localization method based on the same tensor decomposition framework, enabling a fair comparison and highlighting the performance gains introduced by the proposed tracking, data association, and adaptive processing components. Furthermore, we assess the impact of \acp{rrm} deployment by comparing the localization accuracy with and without their inclusion. Performance is quantified in terms of the 2D location \ac{rmse}, 2D \ac{mae}, and the $ y$-axis \ac{mae} (Y-MAE), which is relevant for lane detection.
TeNFiLoc exploits the \ac{los} measurement directly, using the tuple $(\widehat{\phi}_\ell, \widehat{\psi}_\ell, \widehat{\kappa}_\ell\, \text{or}\, \widehat{d}_\ell)$. When the \ac{los} path is not available, it estimates the user position by minimizing the following cost function:
\begin{equation}
    g(\mathbf{p}_u,\widehat{\bm\rho}_\ell) = \sum_{\ell=0}^{L-1} \omega_\ell \left(\|\mathbf{p}_u - \widehat{\mathbf{p}}_{s,\ell}\| - \widehat{d}_\ell + \widehat{\kappa}_\ell \right),
\end{equation}
where $\omega_\ell$ is a weighting factor and $\widehat{\mathbf{p}}_{s,\ell}$ denotes the estimated reflector position of the $\ell$-th path, obtained from \eqref{eq:x_l}. For convergence, the \ac{ls} algorithm requires $L > 3$ as well as perfect synchronization; therefore, it is not considered in the following results. 
Figure~\ref{fig:chan_est} illustrates an example of the estimated scatter locations (including both \ac{ue} and reflectors) derived from the \ac{srnf} channel. The scatter color encodes the altitude, highlighting the spatial distribution of the reflectors and enabling identification of the associated \acp{rrm} based on their elevation. The distribution of the scatters also reflects the measurement quality, particularly for those associated with the \ac{ue} slalom trajectory (blue line).

\begin{figure*}[t]
\centering
\includegraphics[width=0.99\linewidth]{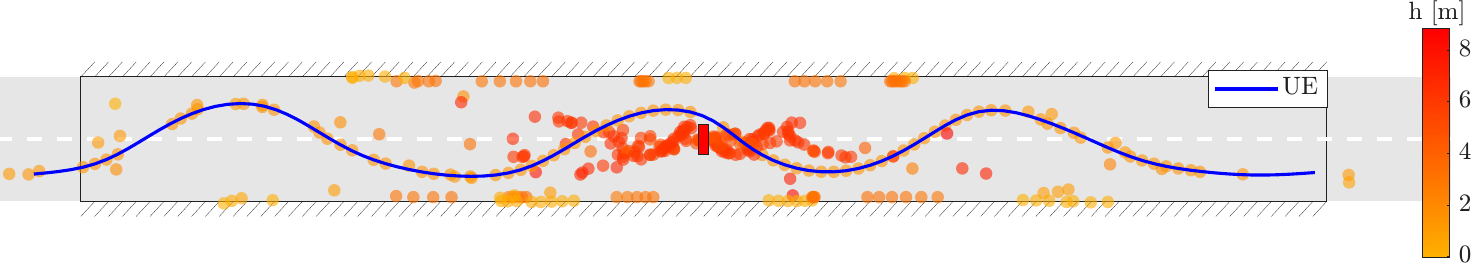}
\caption{UE and reflector position estimates derived from the \ac{srnf} channel. The blue line represents the true UE slalom trajectory, while the red rectangle denotes the BS. The scatter colormap encodes altitude, highlighting the 3D reflector positions: yellow corresponds to bottom RRMs, orange to top RRMs, and red to the ceiling.}
\label{fig:chan_est}
\end{figure*}

\begin{figure*}[t]
\centering
\subfloat[\label{fig:dritto}]{\vspace{-12pt}
\includegraphics[width=0.99\linewidth]{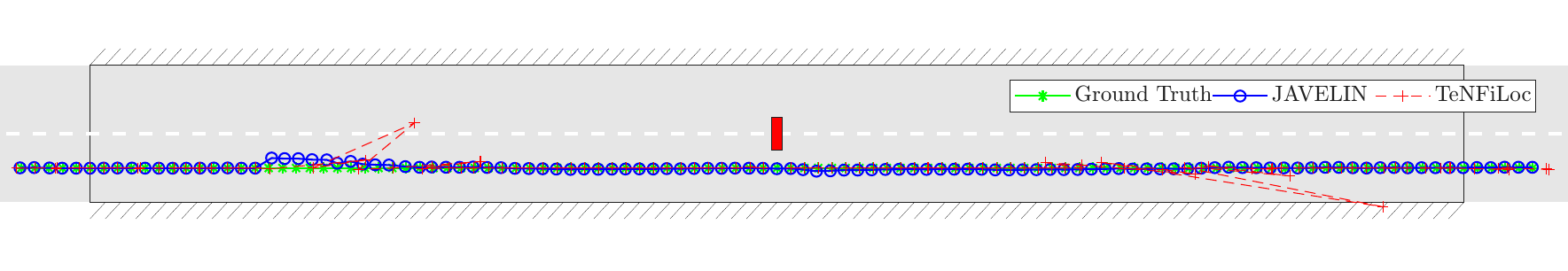}
}\\\vspace{-2pt}
\subfloat[\label{fig:slalom}]{\vspace{-12pt}
\includegraphics[width=0.99\linewidth]{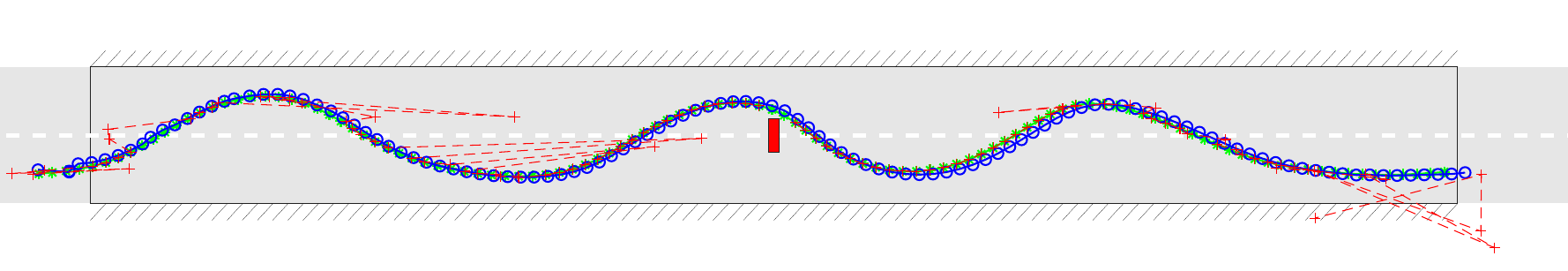}
}
\caption{Estimated vehicle positions by JAVELIN (blue circles) and TeNFiLoc (red crosses) under LoS conditions, for a straight (a) and a slalom (b) trajectory. Ground truth is indicated by green asterisks.}
\label{fig:confronto}
\end{figure*}
Figure~\ref{fig:confronto} compares JAVELIN (blue circles) and TeNFiLoc (red crosses) under \ac{los} conditions for two trajectories, namely a straight path and a slalom. As shown, JAVELIN accurately tracks the vehicle position in both cases. In contrast, TeNFiLoc achieves high accuracy when $\widehat{\kappa}_\ell$ is available; otherwise, it relies on $\widehat{d}_\ell$, which is affected by clock bias, leading to degraded performance.
Table~\ref{tab:straight_comparison} summarizes the performance metrics for the straight trajectory under different \ac{los} conditions: \ac{los} (L), partial \ac{nlos} with a 50\% probability (N@0.5), and complete \ac{nlos} (N). Additionally, the impact of removing \acp{rrm} (NoRRMs) is evaluated. The results show that JAVELIN-L achieves the highest accuracy, with the lowest 2D RMSE, 2D MAE, and Y-MAE, reaching decimeter-level accuracy. As channel conditions degrade towards partial and full \ac{nlos}, JAVELIN experiences a noticeable performance drop, with JAVELIN-N incurring higher errors than its \ac{los} counterpart, while still outperforming TeNFiLoc-L and resolving the location problem with sub-meter precision. In contrast, the absence of \acp{rrm} leads to the largest degradation, emphasizing their key role in maintaining robustness under adverse propagation conditions.
Table~\ref{tab:slalom_comparison} reports the corresponding results for the slalom trajectory, which introduces more dynamic propagation effects. While JAVELIN-L remains the most accurate solution, the overall error levels increase compared to the straight case. This behavior is expected, as the slalom motion causes rapid changes in the environment, making it more challenging to consistently track \acp{vue}. Consequently, the performance gap between \ac{los} and \ac{nlos} conditions becomes more pronounced, with JAVELIN-N showing a significant degradation. Nevertheless, it continues to outperform TeNFiLoc-L in terms of 2D metrics. These findings further highlight the effectiveness of the proposed framework, even in highly dynamic and challenging scenarios.

Figure~\ref{fig:cdf} further corroborates these findings by showing the \ac{cdf} of the 2D localization error for both trajectories and different configurations. JAVELIN-L consistently achieves the best performance, with a steeper \ac{cdf} and a higher concentration of low-error estimates, particularly in the straight trajectory scenario. The performance gap becomes more evident at lower error thresholds, where JAVELIN-L significantly outperforms TeNFiLoc-L. In the slalom case, all methods exhibit a broader error distribution due to the increased dynamics of the environment; however, the proposed framework maintains a clear advantage. The degradation observed for JAVELIN-N and the configuration without \acp{rrm} is also reflected in the heavier tails of their distributions, indicating a higher probability of large localization errors. Overall, these results confirm the effectiveness and robustness of the proposed approach across different propagation conditions and motion patterns, remarking the beneficial effects of \acp{rrm} deployment.
% \begin{table*}[!t]%
% \centering %
% \caption{Performance Metric Comparison.\label{tab:comparison}}%
% \begin{tabular*}{\textwidth}{@{\extracolsep\fill}lllll@{\extracolsep\fill}}
% \toprule
%  & \textbf{JAVELIN LoS Straight}  & \textbf{JAVELIN LoS Slalom}  & \textbf{TeNFiLoc LoS Straight}  & \textbf{TeNFiLoc LoS Slalom} \\
% \midrule
% col1 text & col2 text  & col3 text  & col4 text  & col5 text\tnote{$^\dagger$}   \\
% col1 text & col2 text  & col3 text  & col4 text  & col5 text   \\
% col1 text & col2 text  & col3 text  & col4 text  & col5 text\tnote{$^\ddagger$}   \\
% \bottomrule
% \end{tabular*}
% \begin{tablenotes}
% \item[$^\dagger$] Example for a first table footnote.
% \item[$^\ddagger$] Example for a second table footnote.
% \item {\it Source}: Example for table source text.
% \end{tablenotes}
% \end{table*}

\begin{center}
\begin{table*}[ht]%
\caption{Performance metrics comparison on straight trajectory.\label{tab:straight_comparison}}
\begin{tabular*}{\textwidth}{@{\extracolsep\fill}lccccc@{}}
\toprule
% &\multicolumn{5}{@{}c}{\textbf{straight trajectory}}  \\\cmidrule{2-6}
 & \textbf{JAVELIN-L} & {\textbf{JAVELIN-N@0.5}} & {\textbf{JAVELIN-N}} & \textbf{TeNFiLoc-L} &  \textbf{JAVELIN-NoRRMs}  \\
\midrule
\textbf{2D RMSE [m]} & 0.20  & 0.75 & 0.88  & 2.03  & 2.26 \\
\textbf{2D MAE [m]} & 0.14  & 0.63  & 0.66 & 0.73  & 1.34 \\
\textbf{Y-MAE [m]} & 0.09  & 0.17  & 0.41  & 0.13  & 0.16 \\
\bottomrule
\end{tabular*}
\end{table*}
% \end{center}
% \begin{center}
\begin{table*}[ht]%
\caption{Performance metrics comparison on slalom trajectory.\label{tab:slalom_comparison}}
\begin{tabular*}{\textwidth}{@{\extracolsep\fill}lcccc@{}}
\toprule
% &\multicolumn{4}{@{}c}{\textbf{Slalom Trajectory}}  \\\cmidrule{2-5}
 & \textbf{JAVELIN-L} & {\textbf{JAVELIN-N@0.5}} & {\textbf{JAVELIN-N}} & \textbf{TeNFiLoc-L} \\
\midrule
\textbf{2D RMSE [m]} & 0.41  & 0.56 & 1.38  & 4.76  \\
\textbf{2D MAE [m]} & 0.36  & 0.47  & 1.03 & 1.56  \\
\textbf{Y-MAE [m]} & 0.14  & 0.23  & 0.37  & 0.14  \\
\bottomrule
\end{tabular*}
\end{table*}
\end{center}

\begin{figure}
    \centering
    \includegraphics[width=0.6\linewidth]{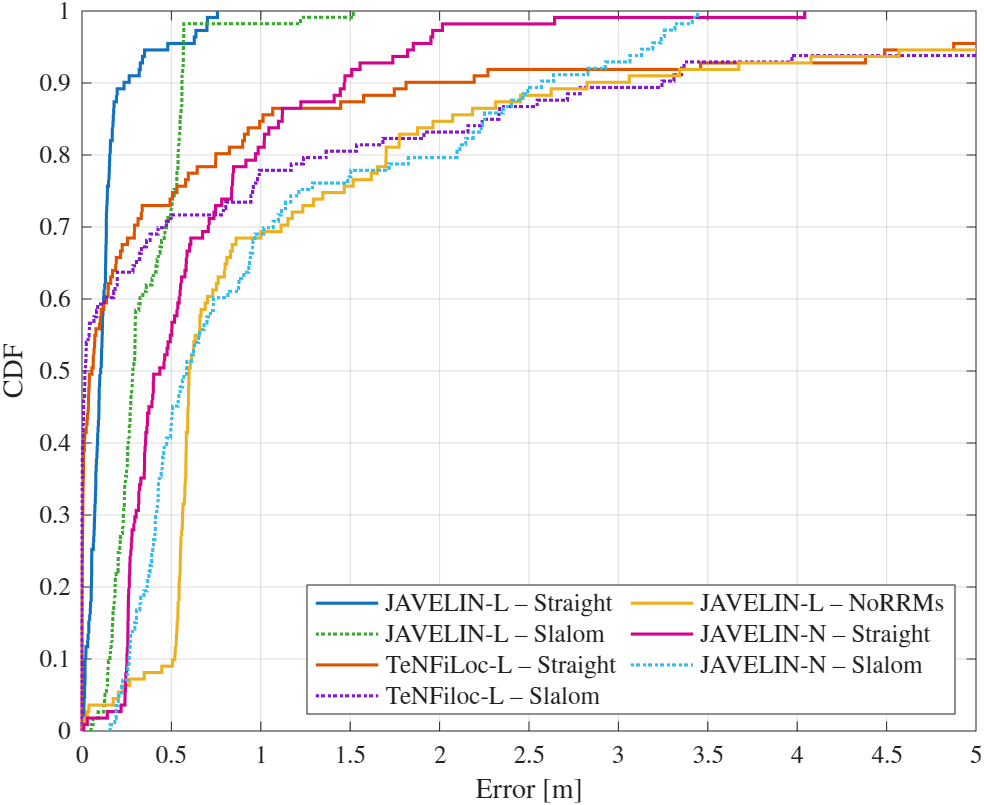}
    \caption{2D localization error CDF for straight and slalom trajectories under different propagation conditions and configurations.}
    \label{fig:cdf}
    %\vspace{-50pt}
\end{figure}

% \begin{center}
% \begin{table*}[!h]%
% \caption{Performance metrics comparison.\label{tab:comparison}}
% \begin{tabular*}{\textwidth}{@{\extracolsep\fill}lcccccc@{}}
% \toprule
% &\multicolumn{4}{@{}c}{\textbf{Straight Trajectory}} & \multicolumn{2}{@{}c}{\textbf{Slalom Trajectory}} \\\cmidrule{2-5}\cmidrule{6-7}
%  & \textbf{JAVELIN-L} & {\textbf{JAVELIN-N@0.5}} & {\textbf{JAVELIN-N}} & \textbf{TeNFiloc-L} &  \textbf{JAVELIN}  & \textbf{TeNFiloc}  \\
% \midrule
% \textbf{2D RMSE [m]} & 0.20  & 0.75 & 0.88  & 2.03  & 0.41 & 3.82 \\
% \textbf{2D MAE [m]} & 0.14  & 0.63  & 0.66 & 0.73  & 0.35  & 1.43 \\
% \textbf{Y-MAE [m]} & 0.09  & 0.17  & 0.41  & 0.13  & 0.14 &  0.24 \\
% \bottomrule
% \end{tabular*}
% % \begin{tablenotes}%%[341pt]
% % \item[$^{\rm a}$] Example for a first table footnote.
% % \item[$^{\rm b}$] Example for a second table footnote.
% % \item {\it Source}: Example for table source text.
% % \end{tablenotes}
% \end{table*}
% \end{center}

\section{Conclusions and Future Work}
\label{sec:conclusions}

This paper investigated a \ac{6g}-oriented approach for single-anchor vehicular localization using the cellular \ac{v2x} technology, exploiting \ac{nf} propagation and passive \acp{rrm} in tunnel environments. We first established an analytical validity condition for the use of an approximated \ac{srnf} channel model, providing an upper bound on the array size at the \ac{bs} beyond which multipath propagation can be consistently interpreted as a single reflector. This result reveals a direct connection between geometric consistency and Fresnel-region scaling, offering important design insights for practical deployments.
Building on this theoretical foundation, we proposed JAVELIN, a single-anchor precise vehicle localization system that combines tensor-based parameter estimation, adaptive \ac{nf}/\ac{ff} processing, and recursive Bayesian tracking with data association and track management. The integration of angular, delay difference, and curvature measurements within a variable-dimension \ac{ekf} enables robust tracking without requiring prior knowledge of the environment.
Simulation results in realistic tunnel scenarios demonstrated that the proposed approach achieves high vehicle localization accuracy under different propagation conditions and motion patterns. In particular, JAVELIN consistently outperforms state-of-the-art single-anchor methods, while maintaining robustness in challenging \ac{nlos} conditions. Furthermore, the introduction of \acp{rrm} was shown to significantly enhance geometric diversity and improve positioning performance, especially in degraded visibility conditions, highlighting their role as a key enabler of scalable and cost-efficient future \ac{c-its} infrastructure.
Overall, the JAVELIN framework aligns with the \ac{6g} vision of positioning-as-a-service embedded in the cellular network, demonstrating how \ac{nf} propagation and smart environments can enable scalable, infrastructure-efficient localization solutions.

Future work will focus on several research directions. First, the extension to real-world experimental validation is a key step to assess the impact of hardware impairments, channel estimation errors, and model mismatches. Second, moving beyond the \ac{srnf} channel model and investigating its impact on the proposed framework is needed to address generalization concerns. Third, the joint optimization of reflector placement and network deployment represents an interesting avenue to maximize localization performance while minimizing infrastructure cost. Additionally, extending the framework to multi-user and cooperative scenarios could enable information sharing among vehicles, further enhancing accuracy and reliability. Finally, integrating emerging \ac{6g} positioning features with additional onboard sensors (e.g., LiDAR or IMU) is envisioned as a unified solution for resilient, seamless, high-precision vehicular localization in complex environments.

%\backmatter
\bmsection*{Author contributions}
Lorenzo Italiano analyzed the literature, designed the methodology, performed the simulation, prepared the figures, and wrote the main manuscript. Mattia Brambilla and Monica Nicoli designed the methodology, analyzed the results, and revised the manuscript. All authors have read and agreed to the published version of the manuscript

\bmsection*{Acknowledgments}
This work was supported by the European Union—NextGenerationEU under the National Sustainable Mobility Center (Grant CN00000023), and by the Italian Ministry of University and Research (MUR) Decree n. 352–09/04/2022.

\bmsection*{Financial disclosure}

None reported.

\bmsection*{Conflict of interest}

The authors declare no potential conflict of interest.

\bibliography{wileyNJD-AMA}

\nocite{*}% Show all bib entries - both cited and uncited; comment this line to view only cited bib entries;

% \bmsection*{Author Biography}

% \begin{biography}{\includegraphics[width=76pt,height=76pt,draft]{empty}}{
% {\textbf{Author Name.} Please check with the journal's author guidelines whether
% author biographies are required. They are usually only included for
% review-type articles, and typically require photos and brief
% biographies for each author.}}
% \end{biography}

\end{document}